\documentclass[a4work]{amsart}

\usepackage{amsmath}
\usepackage{amsthm}
\usepackage{amssymb}
\usepackage[]{graphicx}
\usepackage{adjustbox}
\usepackage{calc} 
\usepackage{txfonts}
\usepackage{enumitem}
\usepackage{ifpdf} 
\usepackage{ifthen}
\usepackage{xcolor}
\usepackage{natbib}
\usepackage{todonotes}
\usepackage{subfigure}
\usepackage{algorithm}
\usepackage{algpseudocode}
\usepackage{tikz}
\usepackage{pgfplots}
\usepackage{bm}
\newcommand{\upartial}[0]{\partial}
\usetikzlibrary{shapes.geometric}

\usepackage[all]{xy}

\newcommand{\R}{\mathbb{R}}

\newcommand{\pprob}[1]{\mathrm{P} \parent{{#1}}}

\newcommand{\transpose}[1]{{#1}^{\mathrm{T}}}

\newcommand{\expf}[1]{\mathrm{exp}\left ( {#1}\right )}
\newcommand{\expfs}[1]{\mathrm{e}^{#1}}

\newcommand{\parent}[1]{\left( {#1} \right)}
\newcommand{\bracket}[1]{\left [ {#1} \right ]}

\newcommand{\absval}[1]{\left| {#1} \right|}
\newcommand{\sset}[1]{\left\lbrace {#1} \right\rbrace }
\newcommand{\arrayh}[1]{\left[ {#1} \right] }

\newcommand{\ssum}[2]{\displaystyle\sum\limits_{#1}^{#2}}

\newcommand{\vvec}[1]{\bm{{#1}}}
\newcommand{\volavec}[0]{\vvec \sigma}
\newcommand{\vvecc}[2]{{#1}_{#2}}
\newcommand{\mmat}[1]{\mathbf{{#1}}}
\newcommand{\mmatc}[3]{{#1}_{#2 #3}}
\newcommand{\ttens}[1]{\bm{{#1}}}
\newcommand{\ddif}[0]{{\rm{d}}}

\newcommand{\bigo}[1]{\mathcal O \parent{{#1}}}

\newcommand{\dbar}[1]{\overline{\overline{{#1}}}}

\newcommand{\startident}{\parent{\vvec x_0}}
\newcommand{\pone}[0]{\mmat{P}_1}
\newcommand{\efvec}[0]{\vvec X}

\newcommand{\eft}[1]{\efvec \parent{{#1}}}
\newcommand{\poneft}[1]{\pone \eft{{#1}}}
\newcommand{\portf}{S_1}
\newcommand{\portft}[1]{\portf \parent{#1}}
\newcommand{\pport}{\prq S^{\startident}}
\newcommand{\pportt}[1]{\pport \parent{{#1}}}

\newcommand{\prq}[1]{\overline{{#1}}}
\newcommand{\pdrift}{\prq a^{\startident}}
\newcommand{\pvola}{\prq b^{\startident}}

\newcommand{\discount}[2]{\expf{-r \parent{{#2}-{#1}} }}
\newcommand{\discountzero}[1]{\expf{-r {#1}}}
\newcommand{\discountzeroreverse}[1]{\expf{r {#1}}}

\newcommand{\stoptimeb}[0]{{\prq \tau}^\dagger }
\newcommand{\atime}[0]{q}

\newcommand{\gmartingale}[0]{R}

\newcommand{\bbt}[0]{{\mmat{b} \transpose{\mmat{b}}}}
\newcommand{\pbbt}[0]{\pone \bbt \transpose{\pone}}

\newcommand{\mathoper}[1]{{\mathop{\mathrm{{#1}}}}}

\newcommand{\expp}[1]{\mathrm{E} \bracket{{#1}}}

\newcommand{\ito}{It\^o\  }

\definecolor{darkblue}{RGB}{0,0,139}
\definecolor{darkgreen}{RGB}{0,139,0}
\definecolor{darkred}{RGB}{139,0,0}
\definecolor{sprgreen}{RGB}{0,250,154}
\definecolor{viol}{RGB}{199,21,133}
\definecolor{royblue}{RGB}{65,105,225}

\definecolor{navy}{RGB}{102,153,255}
\definecolor{tuerkis}{RGB}{51,153,204}

\theoremstyle{plain}
\newtheorem{theorem}{Theorem}[section]
\newtheorem{lemma}[theorem]{Lemma}

\newtheorem{corollary}[theorem]{Corollary}

\theoremstyle{definition}
\newtheorem{definition}[theorem]{Definition}

\newtheorem{assumption}[theorem]{Assumption}
\newtheorem{remark}[theorem]{Remark}

\begin{document}

\author[C. Bayer]{Christian Bayer}
\address{Weierstrass Institute, Mohrenstrasse 39, 10117 Berlin, Germany}
\email{christian.bayer@wias-berlin.de}

\author[J. H\"app\"ol\"a]{Juho H\"app\"ol\"a}
\address{CEMSE, King Abdullah University of Science and Technology (KAUST), Thuwal 23955-6900, Saudi Arabia}
\email{juho.happola@kaust.edu.sa}

\author[R. Tempone]{Raul Tempone}
\address{CEMSE, King Abdullah University of Science and Technology (KAUST), Thuwal 23955-6900, Saudi Arabia}
\email{raul.tempone@kaust.edu.sa}


\title{
Implied Stopping Rules for American Basket Options from Markovian Projection
} 




\begin{abstract}  

This work addresses the problem of pricing American basket options in a multivariate setting, which includes among others, the Bachelier and the Black-Scholes models.
In high dimensions, nonlinear partial differential equation methods for solving the problem become prohibitively costly due to the curse of dimensionality.
Instead, this work proposes to use a stopping rule that depends on the dynamics of a low-dimensional Markovian projection of the given basket of assets.
It is shown that the ability to approximate the original value function by a lower-dimensional approximation is a feature of the dynamics of the system and is unaffected by the path-dependent nature of the American basket option.
Assuming that we know the density of the forward process and using the Laplace approximation, we first efficiently evaluate the diffusion coefficient corresponding to the low-dimensional Markovian projection of the basket.
Then, we approximate the optimal early-exercise boundary of the option by solving a Hamilton-Jacobi-Bellman partial differential equation in the projected, low-dimensional space.
The resulting near-optimal early-exercise boundary is used to produce an exercise strategy for the high-dimensional option,
thereby providing a lower bound for the price of the American basket option.
A corresponding upper bound is also provided.
These bounds allow to assess the accuracy of the proposed pricing method.
Indeed, our approximate early-exercise strategy provides a straightforward lower bound for the American basket option price.
Following a duality argument due to Rogers, we derive a corresponding upper bound solving only the low-dimensional optimal control problem.
Numerically, we show the feasibility of the method using baskets with dimensions up to fifty.
In these examples, the resulting option price relative errors are only of the order of few percent.
   
\end{abstract}

\maketitle

\section{Introduction}

This work addresses the problem of pricing American basket options in a multivariate setting.
Our approach relies on a stopping rule that depends on the dynamics of a low-dimensional Markovian projection of the given basket of assets.

Pricing path-dependent options is a notoriously difficult problem.
Even for relatively simple cases, such as the Black-Scholes model or the Bachelier model, in which an analytic expression of the risk-neutral expected payoff at a terminal time, $T$, can be found, prices of path-dependent options, such as American options, must typically be solved for numerically.
This difficulty is aggravated in high dimensions, where convergence rates of well-known numerical methods deteriorate exponentially as the number of dimensions increases.
However, there is a plethora of American options being offered in the markets, in publicly traded markets or over-the-counter (OTC).
Perhaps the best-known example is that of options written the S{\&}P-100 index quoted on the Chicago Board Options Exchange (CBOE).
In addition, the wide variety of exchange traded funds (ETF) tracking indices have American options written on them publicly quoted on CBOE.
These funds include many prominent indices such as Euro Stoxx 50 and the Dow Jones Industrial average, as well as many regional indices.
If one is interested in the index alone, then 
a low-dimensional model for the index is clearly sufficient. However, in many situations, consistent joint models of the index together with some or all the individual stocks may be required, which would lead 
to the moderate and high dimensional option pricing problems addressed in this paper.

The two most widely used approaches to pricing path-dependent options, binomial tree methods and partial differential equation (PDE) methods, both suffer from the so-called curse of dimensionality.
In the case of the probability trees or lattices, the size of the probabilistic trees, even in the case of recombining trees, already becomes prohibitively large in moderate dimensions.
The other popular method requires solving the Black-Scholes equation using finite difference (FD) or finite element (FEM) methods.
Both methods involve discrete differential operators whose size also scales exponentially in the number of dimensions.

In Monte Carlo simulation, the rate of convergence of weak approximations does not explicitly depend on the number of dimensions.
With early-exercise options like American ones, however, Monte Carlo methods become more complicated. 
Although well suited for forward-propagation of uncertainties in a wide range of models, traditional Monte Carlo methods do not offer a straightforward way to construct an exercise strategy.
Such a strategy typically needs to be obtained through backward induction.
Because the price of an American option is based on assuming optimal execution of the option, any solution scheme needs to produce the optimal stopping strategy as a by-product of the pricing method.
Many methods have been developed to produce a near-optimal execution strategy.
\citet{broadie1997pricing} introduced a pair of schemes that evaluate upper and lower bounds of the prices of American options.
\citet{longstaff2001valuing} used least-squares regression in conjunction with Monte Carlo simulation to evaluate the price of American options. Their popular method has been widely implemented in various pricing engines, for example in the QuantLib library by \citet{Ame2003}.

In the least-squares Monte Carlo methodology, the value of holding an option is weighed against the cash flow captured by exercising the option.
The intrinsic value of an option is, of course, known. However, the holding price is the discounted expectation of possible future outcomes.
This expectation is estimated based on a Monte Carlo sample by regressing the holding price of the option to a few of decision variables or basis functions.
Naturally, the choice of the appropriate basis functions has a crucial effect on the quality of the outcome, and also the number of basis functions should be much smaller than the size of the Monte Carlo sample to avoid overfitting\citep{glasserman2004number,zanger2013quantitative,zanger2016convergence}.
For work on the reduction of the computational complexity in the regression methods, we refer the reader to \citet{belomestny2015pricing}.

Another method to approximate option prices in high dimensions is the optimal quantizer approach of \citet{bally2005quantization}.
In this method the diffusion process is projected to a finite mesh. This mesh is chosen optimally to minimize projection error, the conditional expectation describing the holding price is then evaluated at each of the mesh points.
The quantization tree approach gives accurate approximations of the option price in moderate dimension.
Here, we present methods for selected parametrisations of the Black-Scholes model over twice the dimension presented in \citep{bally2005quantization} For work with rather large number of dimensions, we refer the reader to the stratified state aggregation along payoff (SSAP) method of \citep{barraquand1995numerical}.
In the SSAP method, one solves for an exercise strategy through stratifying possible values of the intrinsic value of the option.
\citet{andersen1999simple} used a similar approach for pricing Bermudan swaptions, characterizing the early exercise boundary in terms of the intrinsic value.

Here, we propose and analyze a novel method for pricing American options written on a basket of assets.
Like the SSAP, the pricing method in this work relies on using the intrinsic value, or the value of the underlying asset as a state variable. 
On the other hand, our method is based on the Markovian projection of the underlying asset, does not rely on the use of basis functions
and provides upper and lower bounds for the option price.
These bounds are useful to assess the accuracy of our methodology.

In this exploratory work, we computationally study the feasibility of using stopping rules based on a simplified surrogate process in pricing American options written on a basket of assets.
The method offers an efficient approximation to pricing and hedging American options written on an index, or a security tracking such index.
Instead of the full-dimensional process, we use a lower-dimensional process obtained through Markovian projection.
Even though the evolution of the multiple assets involved in a given basket is usually assumed Markovian,
the SDE describing the evolution of a linear combination or a basket of assets, is rarely Markovian in the basket value.
We address this issue by means of Markovian projection, which provides a low-dimensional Markovian SDE that is suited to dynamic programming (DP) methods that solve the relevant Hamilton-Jacobi-Bellman (HJB) equation.
Markovian projection techniques have been previously applied to a range of financial applications, see, for example, \citep{piterbarg2003stochastic,piterbarg2005stochastic,djehiche2014risk}.

\paragraph{Outline}
The remainder of this work is organized as follows.
In Section \ref{sec:MarPro}, we describe the Markovian projection in the context of projecting high-dimensional SDEs into lower dimensions.
We show how the low-dimensional HJB equation gives rise to a stopping rule that in general is sub-optimal but provides a lower bound for the American option price.
Using a duality approach from \citet{rogers2002monte}, we give an upper bound for the option price using the solution of the low-dimensional HJB equation.
We show that in the Bachelier model, the lower and upper bounds coincide and provide an exact option valuation.
We prove how the question of whether the cost-to-go function of an American option can be approximated using a low-dimensional approximation reduces to the corresponding question of European options, which are simpler to analyze.
It is known that the Bachelier model is a close approximation to the Black-Scholes model in the realm of European option pricing \citep{schachermayer2008close}.
We motivate that this approximation has a beneficial effect when pricing American basket options with our methodology since our method is exact for the Bachelier model.
In Section \ref{sec:implementation}, we detail the numerical implementation of the ideas developed in the preceding section and experiment with multivariate Bachelier and Black-Scholes models.
Reporting results of numerical experiments, we verify the accuracy of our method with the Bachelier model and give supporting results to justify the use of our method in cases where neither the European or American option prices can be precisely represented using a low-dimensional approximation.
Using the Black-Scholes model as an example, we show that the approximation error of our method is few per cent, comparable to the bid-ask spread of even the more liquid openly traded options and well within the spread of more illiquid index options or options quoted on an ETF.
Finally, we offer concluding remarks in Section \ref{sec:conclusions}.

\section{Markovian projections and implied stopping times}
\label{sec:MarPro}

In this section, we revisit the essential equations that describe risk-neutral option pricing of American options in a multivariate setting.
We present in Section \ref{subsec:markovianbasics} how these equations have corresponding low-dimensional projections that can be obtained using the Markovian projection.
In Section \ref{subsec:bounds}, we show how the projected PDEs give rise to lower and upper bounds for the solution of the original high-dimensional pricing problem.

Following the introduction of the relevant bounds, we discuss in Section \ref{subsec:DimRed} 
classes of models that are of particular interest in reduced-dimension evaluation.
First, we recall in Lemma \ref{lem:bachelierReduction} how the Gaussian Bachelier model has the feature that the Markovian projection produces a one-dimensional SDE whose solution coincides in law with the underlying high-dimensional portfolio.
We also show in Corollary \ref{cor:stochclock} how this one-dimensional approximation property is preserved if the Bachelier model is generalized through the appropriate introduction of a stochastic clock.
Secondly, we provide auxiliary results to characterize some \ito SDEs that have this exact reduced dimension structure that our proposed method exploits.
Among these ancillary results, we have Lemma \ref{lem:european}, which we use to reduce the discussion of dimension reduction of American options into the problem of analyzing low-dimensional approximations of the corresponding European option.
Furthermore, we give a motivation for using the Markovian projection even for models that do not have the exact reduced dimension property.

\subsection{Markovian projections and approximate stopping times}

\label{subsec:markovianbasics}

Assume that the time evolution of the 
asset prices in the basket is given by a stochastic process in $\R^d$, $\efvec(t,\omega)$,
that is the unique strong solution to
an It\^o SDE,
\begin{align}
\begin{split}
\ddif \eft t &= 
\vvec a \parent{t,\eft t} \ddif t + \mmat{b} \parent{t,\eft t} 
\ddif \vvec W \parent t, ~~~~0< t <T,
\\
\eft 0 &= \vvec x_0,
\end{split}
\label{eq:firstFullSDE}
\end{align}
which is driven by a $k$-dimensional Wiener process with independent components, $\vvec W$. 
We work under the risk-neutral measure and due to a no arbitrage assumption, the drift in \eqref{eq:firstFullSDE} is a linear function,
\begin{align}
\vvec a \parent{t, \vvec x} = r \vvec x,
\label{eq:riskNeutralDrift}
\end{align}
where $r \in \mathbb R$ is the short rate.
Most of the discussion can also be generalized with minimal modifications to a time-dependent, stochastic, short rate when the short rate process is independent of the dynamics of the underlying assets, see Remark \ref{rem:stochIntRates}.
For $ 1\le i\le d$ and $1\le j\le k$ the diffusion coefficients, 
$\mmatc b i j $,
are at least second order differentiable functions and such that the pdf of $\eft t$ exists for $0<t\le T$ and is a univariate, smooth function, cf. Assumption \ref{ass:knownDensity}.
Furthermore, we denote the canonical filtration generated by $\eft t $ as
$
\mathcal F_t = \sigma \sset{
{\eft \atime } : 0 \leq \atime \leq t}.
$

In the numerical examples in the subsequent section, we directly deal with the models of Bachelier \citep{sullivan1991louis} and Black-Scholes \citep{black1973pricing}, acknowledging possible extensions to the constant elasticity of variance (CEV) model (see \citet{cox1975notes}) that can in a certain sense be understood as a compromise between the Bachelier and Black-Scholes models.
Many other extensions are also possible, and we discuss some of them in Section \ref{subsec:DimRed}.
Note the time-homogeneous structure of the examined models and recognize possible extensions to time-inhomogeneous models, for instance by using temporal reparametrization. 

Furthermore, we assume for simplicity that the underlying pays no dividends.
This work focuses extensively on models of Bachelier and Black-Scholes type.
They are defined by their respective volatilities, namely
\begin{align}
\mmat b_{\text{Bachelier}} \parent{t,\vvec x} &= \mmat \Sigma ,
\label{eq:BachelierDefinition}
\\
b_{\text{Black-Scholes},ij} \parent{t,\vvec x} &= \vvecc x i \Sigma_{ij} .
\label{eq:BlackScholesDefinition}
\end{align}
with $\mmat \Sigma \in \mathbb R^{d\times k}$ in both models.

We focus on a portfolio of assets, $\portf$, given by weights $\pone$,
\begin{align}\label{eq:Sdef}
\portft {t} = \poneft t,
\end{align}
as the underlying security, for $\pone \in \R^{1 \times d}$, with non-zero elements, possibly some but not all negative.
We seek to price options with the payoff functional $g : \mathbb R \rightarrow \mathbb R$.
Arguably, the most interesting example is that of the put option, $g \parent {s} = \parent{K-s}^+$ for some $K \in \mathbb R$.

The price of the European option written on the portfolio $\pone$ with expiry at $T$ is given by
\begin{align}
\label{eq:europeanFullPde}
u_E \parent{t,\vvec x} =\expp{ \discount t T g \parent{\poneft T } | \eft t = \vvec x}.
\end{align}

In contrast, when pricing American options, we seek to solve for
\begin{align}
\begin{split}
u_A \parent{t,\vvec x} &= \mathop{\mathrm{sup}}_{\tau \in \mathcal T_t} \expp{ \discount t \tau g \parent{\poneft \tau} |\eft t = \vvec x},
\\
\mathcal T_{\atime} &= \sset{\tau : \Omega \rightarrow [\atime ,T] | 
\sset{ \tau \leq t} 
\in \mathcal F_t,~~ \forall t \in [\atime,T]}.
\end{split}
\label{eq:americanPriceDefinition}
\end{align}
The European option price $u_E$ given by \eqref{eq:europeanFullPde} also satisfies the Black-Scholes equation in \mbox{$\parent{t, \vvec x} \in [0,T]\times D$,}
\begin{align}
\begin{split}
-\upartial_t u_E \parent{t,\vvec x} 
=& -r u_E \parent{t,\vvec x} 
 +\ssum{i}{} \vvec a_i \parent{t, \vvec x} \upartial_{\vvecc x i} u_E \parent{t,\vvec x} + \frac{1}{2}\ssum{ij}{} \parent{\mmat b \transpose{\mmat b}}_{ij}\parent{t,\vvec x} \upartial^2_{\vvecc x i \vvecc x j} 
u_E \parent{t,\vvec x} 
\\
u_E \parent{T,\cdot} =& g \parent{\pone \cdot},
\end{split}
\label{eq:theFullDimensionEuropeanEquation}
\end{align}
with the appropriate domain $D \subset \mathbb R^d$. For example, in the Black-Scholes model, we have \mbox{$D = D_{BS}^d = \mathbb{R}^d_+$} with the appropriate Dirichlet boundary condition at hyperplanes at which one or more components of $\eft t$ are zero.
The boundary value is given by a lower-dimensional version of \eqref{eq:theFullDimensionEuropeanEquation}.
Defining the second order linear differential operator
\begin{align*}
\parent{\mathcal L  v } \parent{t,\vvec x}
=
\parent{-r + \ssum{i}{} \vvec a_i \upartial_{\vvecc x i} + \frac{1}{2} \ssum{ij}{} \parent{\mmat b \transpose{\mmat b}}_{ij} \upartial^2_{\vvecc x i \vvecc x j} } \parent{t,\vvec x}
v \parent{t,\vvec x} ,
\end{align*}
we can write the corresponding non-linear HJB equation. Following the presentation of \citet[Equation (6.2)]{achdou2005computational},
the American option price, $u_A$, satisfies
\begin{align}
\parent{\mathcal{L} u_A + \upartial_t u_A} \parent{t, \vvec x} &\leq 0,
~~ &\parent{t, \vvec x} \in [0,T] \times D,
\nonumber
\\
u_A \parent{t, \vvec x} &\geq g \parent{\pone \vvec x},
 ~~&\parent{t, \vvec x} \in [0,T] \times D,
 \nonumber
 \\
 \parent{\parent{\mathcal{L} u_A + \upartial_t u_A} \parent{t, \vvec x}} \parent{u_A \parent{t,\vvec x}-g \parent{\pone \vvec x}}
 &= 0,
 ~~&\parent{t, \vvec x} \in [0,T] \times D.
 \nonumber
\end{align}
Introducing the Hamiltonian, 
\begin{align}
\parent{\mathcal H u_A} \parent{t, \vvec x} = \parent{\mathcal L u_A } \parent{t, \vvec x} \mathbf 1_{\mathoper{max} \parent{ \parent{\mathcal L u_A}\parent{t,\vvec x}, ~ u_A \parent{t, \vvec x} - g \parent{\pone \vvec x}} > 0} ,
 \label{eq:hamiltonian}
\end{align}
we write the HJB equation for $u_A$ shortly as
\begin{align}
\begin{split}
-\upartial_t u_A \parent{t, \vvec x} & = \parent{ \mathcal H u_A } \parent{t, \vvec x}, ~~ \parent{t, \vvec x} \in [0,T] \times D,
\\
u_A \parent{t, \cdot} &= g \parent{\pone \cdot}. 
\end{split}
\label{eq:fullAmericanEquation}
\end{align}
For the Bachelier model, $D$ is unbounded.
For the Black-Scholes model, one or more components of $\eft t$ vanish at the boundary $\upartial D$.
Since both the drift \eqref{eq:riskNeutralDrift} and the volatility \eqref{eq:BlackScholesDefinition} are linear in their arguments, the drift and the volatility vanish at the boundary.
Resulting boundary value is thus given by a lower-dimensional variant of \eqref{eq:fullAmericanEquation} where one or more of the components of $\eft t$ are fixed to zero.

Instead of trying to solve \eqref{eq:fullAmericanEquation} directly, 
we first turn our attention to a low-dimensional approximation of the portfolio process $\portf$ introduced in \eqref{eq:Sdef}.
This approximation is the Markovian projection of $\portf$ \citep{gyongy1986mimicking,piterbarg2006markovian}.
Indeed, we approximate the non-Markovian evolution of $\portf$
by the following surrogate process,
\begin{equation}
\begin{split}
\ddif \pportt t =& 
\pdrift \parent{t,\pportt t} \ddif t 
+ \pvola \parent{t,\pportt t} \ddif W \parent t , ~~~~ t \in [0,T],
\\
\pportt 0 =& \pone \vvec x_0, 
\end{split}
\label{eq:surrogate}
\end{equation} 
The drift and volatility coefficients in \eqref{eq:surrogate} are evaluated through conditional expectations, namely
\begin{align}
\pdrift \parent{t, s} 
=& \expp{\pone a \parent{t,\eft t} |\poneft t = s,~\eft 0 = \vvec x_0 },
\label{eq:baradef}
\\
\parent{\pvola}^2 \parent{t, s} 
=& \expp{ \parent{ \pone \mmat b \transpose{\mmat b } \transpose{\pone}} \parent{t, \eft t} |\poneft t = s, ~\eft 0 = \vvec x_0 }.
\label{eq:barbdef}
\end{align}
The Markovian projection \eqref{eq:surrogate} generates its canonical filtration,
$
\prq {\mathcal F}_t = \sigma \sset{\pportt \atime : 0 \leq \atime \leq t}.
$

Observe that the surrogate process, $\pportt t$ in \eqref{eq:surrogate}, has, due to the proper selection of the drift and volatility functions and the appropriate initial value, the same marginal density as \mbox{$S \parent t = \poneft t $ for all $t \in [0,T]$} \citep{gyongy1986mimicking}. For any given payoff function $g$ that yields a finite price in \eqref{eq:europeanFullPde}, this implies the identity 
\begin{align}
\expp{ \expf{-rT} g \parent{\poneft T} | \eft 0 = \vvec x_0}
=
\expp{ \expf{-rT} g \parent{\pportt T} | \pportt 0 = \pone \vvec x_0},
\label{eq:coincidingExpectations}
\end{align}
which means that we can price European options on the basket using only our knowledge of the Markovian process $\pport$.

Assuming that we know the dynamics \eqref{eq:surrogate}, we can evaluate the right-hand side of
\eqref{eq:coincidingExpectations} using the Feynman-Kac Formula.
By denoting
\begin{align}
\label{eq:valueFunctionDefinition}
\prq u_E \parent{t,s} = \expp{\expf{-r\parent{T-t}} g \parent{\pportt t} | \pportt t=s},
\end{align}
we have that $\prq u_E $ solves a corresponding linear backward PDE in one space dimension only, 
\begin{align}
\begin{split}
- \upartial_t \prq u_E \parent{t,s} &= \underbrace{-r \prq u_E \parent{t,s} 
+ \pdrift \parent{t,s} \upartial_{s} \prq u_E \parent{t,s} + \frac{\parent{\pvola}^2 \parent{t,s}}{2} \upartial^2_{ss} \prq u_E \parent{t,s}}_{\equiv \parent{ \prq{ \mathcal L} \prq u_E}\parent{t,s}},
 ~~ t \in [0,T], ~ s \in \prq D,
\\
\prq u_E \parent{T,\cdot} &= g \parent{ \cdot}.
\end{split}
\label{eq:theReducedEuropeanEquation}
\end{align}

\begin{remark}[Interpretation of projected PDEs]
We have defined the projected PDE \eqref{eq:theReducedEuropeanEquation} that is of Black-Scholes type.
Furthermore, the coefficients $\pdrift$ and $\pvola$ of the equation are constructed through conditioning to the initial value of the SDE \eqref{eq:firstFullSDE}.
Here, we use the the PDE \eqref{eq:theReducedEuropeanEquation} as a mathematical construct to evaluate the expectation \eqref{eq:coincidingExpectations}.
We do not interpret the solution of \eqref{eq:theReducedEuropeanEquation}, or its extensions defined in the remainder of this work as tradeable option prices.
\end{remark}

Note that the procedure above can be generalized to cases where the Markovian projection is carried out onto a space of dimension $\prq d >1.$
This is done simply by introducing additional portfolios and their weights,
$\transpose{\mmat P} = [\transpose{\mmat P_1},\transpose{\mmat P_2},\transpose{\mmat P_3},\dots,\transpose{\mmat P_{\prq d}}]$,
and defining the multidimensional dynamics for $\pport$ via the projected volatility coefficients as
\begin{align}
\label{eq:highDimProjection}
\parent{\mmat b \transpose{\mmat b}}^{\startident}_{ij} \parent{t, \vvec s} 
=& \expp{\parent{\transpose{\mmat{P}_i} \mmat b \transpose{\mmat b} \mmat{P}_j} \parent{t, \eft t} | \mmat P \eft t = \vvec s, ~ \eft 0 = \vvec x_0}, ~1\le i,j\le \prq d 
.
\end{align}
Summing up, as long as we can efficiently evaluate the coefficients in the SDE \eqref{eq:surrogate}, it is possible to solve the low-dimensional Equation \eqref{eq:theReducedEuropeanEquation} instead of Equation \eqref{eq:theFullDimensionEuropeanEquation} that suffers from the curse of dimensionality. Obviously, the efficient evaluation of the coefficients in the SDE of $\pport$ via conditional expectation as in \eqref{eq:barbdef} is in principle a daunting task.
Section \ref{subsubsec:laplace} proposes an efficient approximation to carry out this evaluation.

\begin{remark}[Computational domains and boundary conditions]
\label{rem:boundary}
Instead of using the full unbounded domain of the PDE \eqref{eq:theReducedEuropeanEquation} in the numerical part of this work, we use a modified, computational domain, on which we impose an artificial boundary condition as follows.

First, note that the appropriate domain, $D$, for \eqref{eq:theFullDimensionEuropeanEquation} depends on the model of choice.
For the $d$-dimensional Black-Scholes model, we have $D=D^d_{\mathoper{Black-Scholes}}=\mathbb R^d_+$ and correspondingly for the Bachelier model, $D=D^d_{\mathoper{Bachelier}}= \mathbb R^d$.
When numerically solving the full, $d$-dimensional Equation \eqref{eq:theFullDimensionEuropeanEquation}, one often truncates the domain into a compact one and imposes artificial boundary conditions on the boundary of the localized computational domain.
Here, we also truncate the projected domain, $\prq D$, into a localized computational domain.
At the boundary of the computational domain, we impose the artificial boundary condition $\bar u \parent{t,s} = g \parent{s}$.
In addition to the truncation, we note that the coefficients in \eqref{eq:theReducedEuropeanEquation} are defined only for regions where the density $\phi$ of process $\poneft t$ has support.
We extend artificially the domain of \eqref{eq:theReducedEuropeanEquation} to the rectangle $[0,T] \times [s_{\mathoper{min}}, s_{\mathoper{max}}]$ by extrapolating the relevant coefficients $\pdrift$ and $\parent{\pvola}^2$.
For $\parent{\pvola}^2$ we also set a lower bound to guarantee numerical stability and well-posedness.

In all our numerical examples, we make sure that our truncated and extrapolated computational domain is sufficiently large to make the corresponding domain truncation error negligible.
For more in-depth discussions on this matter, we refer the reader to \citep{kangro2000far, choi2001numerical, matache2004fast, hilber2004sparse}.

Furthermore, to maintain brevity of notation, we will refrain from writing explicitly the artificial boundary conditions. All relevant PDEs in this work are understood to be numerically solved using Dirichlet boundary conditions implied by the intrinsic value of the option.
\end{remark}

Just as the Black-Scholes equation, \eqref{eq:theFullDimensionEuropeanEquation} has a corresponding HJB equation \eqref{eq:fullAmericanEquation}, we may use the corresponding HJB to the projected Black-Scholes equation \eqref{eq:theReducedEuropeanEquation}.
The resulting HJB equation describes the cost-to-go function $\prq u_A$ of an American option written on the portfolio that has the projected dynamics of \eqref{eq:surrogate}:
\begin{align}
\begin{split}
- \upartial_t \prq u_A \parent{t,s} &=
\parent{\prq {\mathcal L} \prq u_A } \parent{t, s} 
 \mathbf 1_{\mathoper{max} \parent{ \parent{\prq{\mathcal L} \prq u_A}\parent{t, s}, ~ \prq u_A \parent{t, s} - g \parent{s}} > 0} = \parent{\prq{\mathcal H} \prq u_A} \parent{t, s} 
~~~ \parent{t,s} \in [0,T] \times \prq D,
 \\
\prq u_A \parent{T,\cdot} &= g \parent{\cdot}. 
\end{split}
\label{eq:projectedAmericanBackwardEquation}
\end{align}
However, for American option prices, there is no
identity corresponding to equality \eqref{eq:coincidingExpectations}. 
As a result, the magnitude of the difference $ \absval{ \prq u_A \parent{0, \pone \vvec x_0}- u_A \parent{0 , \vvec x_0} }$ may not necessarily be small.
Also, the boundary conditions in \eqref{eq:projectedAmericanBackwardEquation} are subject to the same ambiguity as the ones of \eqref{eq:theReducedEuropeanEquation} discussed in Remark \ref{rem:boundary}.
The main focus of this work is to address these issues and to estimate the difference between the computed value of $\prq u_A$ and the sought $u_A$, which is assumed beyond our reach being too costly to compute.

We note in passing that the processes $\efvec$ and $\pport$ 
live in different probability spaces. Likewise, the stopping times corresponding to the full-dimensional and projected SDE are adapted to $\mathcal F_t$ and $\prq{\mathcal F_t}$, respectively.

\subsection{Implied stopping time and price bounds}
\label{subsec:bounds}

Above we have laid out the question of the feasibility of using the projected dynamics $\pport $ in pricing American options, we now show below in Section \ref{subsubsec:lbound} how the solution of the projected problem $\prq u_A$ gives rise to an exercise strategy that is sub-optimal.
This sub-optimal exercise strategy gives a lower bound for the option price.
We complement this lower bound with a corresponding upper bound in Section \ref{subsubsec:ubound}.

\subsubsection{Lower bound}
\label{subsubsec:lbound}

In the full American option pricing problem \eqref{eq:americanPriceDefinition}, the optimal stopping time, $\tau^* \in \mathcal T_0$, such that
\begin{align*}
u_A \parent{0,\vvec x_0} = \expp{\discountzero{\tau^*} g \parent{\poneft {\tau^*}} | \eft 0 = \vvec x_0},
\end{align*}
is given by 
\begin{align}
\tau^* = \mathrm{inf} \sset{t \in [0,T]: u_A \parent{t, \eft t} = g \parent{\poneft t}}.
\label{eq:optimalStoppingTime}
\end{align}
Any stopping time $\tau \in \mathcal T_0$ gives a lower bound for the option price.
We do not have access to the full cost-to-go function, $u_A$, and hence a natural replacement is given by the projected cost-to-go function $\prq u_A$. Indeed, the projected cost-to-go function $\prq u_A$ gives rise to two hitting times:
\begin{align*}
\prq{\tau}^* \equiv \mathrm{inf} \sset{t \in [0,T]: \prq u_A \parent{t,\pportt t} = g \parent{\pportt t} },
\end{align*}
where the dynamics of $\prq S$ is given by \eqref{eq:surrogate} and
\begin{align}
\label{eq:hittingTime}
\stoptimeb \equiv \mathrm{inf} \sset{t \in [0,T]: \prq u_A \parent{t,\poneft t} = g \parent{\poneft t}}.
\end{align}
We note that due to the terminal condition on $\prq u_A$ in  \eqref{eq:projectedAmericanBackwardEquation} all hitting times are bounded by $T$. 

We conclude the discussion on the lower bound of the option value by stating the lower bound implied by the hitting time $\stoptimeb \in \mathcal T_0$,
\begin{align}
\label{eq:lowerBoundEquation}
u_A \parent{0,\vvec x_0} \geq \expp{ \discountzero{\stoptimeb} g \parent{\poneft {\stoptimeb}} | \eft 0 = \vvec x_0}.
\end{align}
We emphasize that we have not made a comparison between $u_A \parent{0,\vvec x_0}$ and $\prq u_A \parent{0, \pone \vvec x_0}$.

\begin{remark}[On least-squares Monte Carlo]
The approach we have adopted shares some similarities with the least-squares Monte Carlo approach. However, there are key differences:
In the least-squares Monte Carlo method, the stopping time can be understood as a hitting time into a region where the holding value of the option, as estimated through regression to a set of basis functions, is exceeded by the early exercise price.
The hitting time \eqref{eq:hittingTime} is likewise defined as a comparison between the estimated cost-to-go function, $\prq u_A$, and the early exercise price.
However, the estimated cost-to-go function, $\prq u_A$, does not depend on a choice of basis functions, only on the direction of the projection.
On the other hand, $\prq u_A$ is constructed using the Markovian projection $\pport$ instead of the true forward model $\efvec$.
\end{remark}

\subsubsection{Upper bound}
\label{subsubsec:ubound}

To assess the accuracy of approximating the process with a low-dimensional Markovian projection, we want to devise a corresponding upper bound. For this, we use the dual representation due to \citet{rogers2002monte}.

The dual representation of the pricing problem is as follows.
The price of the American option is given by:
\begin{align}
u_A \parent{0,\vvec x_0} = \mathop{\mathrm{inf}}_{\gmartingale \in H_0^1} \expp{\sup_{0 \leq t \leq T} \parent{\tilde Z \parent t - \gmartingale \parent t}| \eft 0 = \vvec x_0},
\label{eq:dual}
\end{align}
where $H_0^1$ denotes the space of all integrable martingales $\gmartingale$, $t \in [0,T]$
such that for $\gmartingale \in H_0^1$
\begin{equation*}
\begin{split}
\mathoper{sup}_{0 \leq t \leq T} \absval{\gmartingale \parent t} & \in L^1,
\\
\gmartingale  \parent 0 &= 0.
\end{split}
\end{equation*}
Here $\tilde Z \parent t$ denotes the discounted payoff process
\begin{align}
\tilde Z \parent t = 
\discountzero t
g \parent{\eft t},
~~ t \in [0,T].
\label{eq:discountedPayoff}
\end{align}

Naturally, evaluating the statement within the infimum of Equation \eqref{eq:dual} with any martingale, $\gmartingale \parent t \in H_0^1$, will give an upper bound to the option price.
A martingale, $\gmartingale ^* \parent t$, reaching the infimum \eqref{eq:dual} is called an optimizing martingale.
In general, finding an optimizing martingale is as complex as finding the solution to the pricing problem.
In fact, when the cost-to-go function, $u_A \parent{t,\vvec x}$, is known, the optimizing martingale can be written out following the approach in \citet{haugh2004pricing}:
\begin{align}
\begin{split}
\ddif \gmartingale ^* \parent t &= 
\discountzero t
\parent{\transpose{\parent{\nabla u_A}} \mmat b} \parent{t,\eft t} \ddif \vvec W\parent t, ~~~ t \in [0,T],
 \\
\gmartingale ^* \parent 0 &= 0 .
\end{split}
\label{eq:optimalMartingale}
\end{align}
We construct a near-optimal martingale $\gmartingale^\star \in H_0^1$ by replacing in \eqref{eq:optimalMartingale} the exact $u_A$ with the approximate cost-to-go function, $\prq u_A.$ 
This yields the explicit upper bound
\begin{align}
u_A \parent{0,\vvec x_0} \le \expp{\sup_{0 \leq t \leq T} \parent{\tilde Z \parent \atime -\gmartingale^\star \parent t}| \eft 0 = \vvec x_0},
\label{eq:dualapprox}
\end{align}
where
\begin{align}
\begin{split}
\ddif \gmartingale ^\star \parent t &= 
\discountzero t
\parent{ \transpose{\parent{\nabla \prq u_A}} \parent{t, \poneft t}} \pone \mmat b \parent{t, \eft t} \ddif \vvec W \parent t, ~~~ t \in [0,T],
 \\
\gmartingale ^\star \parent 0 &= 0.
\end{split}
\label{eq:subOptimalMartingale}
\end{align}
In other words, we evaluate the sensitivity, or delta, of the projected, approximate value function using the projected, non Markovian, version of the true stochastic process. We also note that the sensitivity of the projected, approximate value function can be used as an approximate sensitivity of the option value with regard to the value of the underlying portfolio. \citep[chapter 3]{rogers2002monte}

\subsection{Dimension reduction for models relevant to quantitative finance}
\label{subsec:DimRed}

We have established a lower as well as an upper bound for the American basket option prices using Markovian projection.
The question of which models feature tight bounds is naturally of interest for the applicability of our methodology.
Thus, this section focuses on the domain of applicability of the Markovian projection.
Below, we demonstrate that the procedure of Markovian projection produces exact results for the Bachelier Model.
This is a consequence of the Gaussian returns in the model.
In fact, it turns out that due to the constant volatility \eqref{eq:BachelierDefinition} of the Bachelier model, the coefficients of the relevant low-dimensional PDEs can be evaluated without Laplace approximation.

Following our discussion about the Bachelier model, we then concentrate on the Black-Scholes model, which is known to produce option prices that are well approximated by the Bachelier model. Finally, we state conditions under which the Black-Scholes model also satisfies the property that the value function of the option depends only on a single state variable $s$, namely the portfolio value $\pone \vvec x$ .

\subsubsection{Definitions}

First, let us define some terminology.
Let $1\le n <d$ and $D \subset \mathbb R^d$ be a convex set with piecewise smooth boundary.

\begin{definition} 
We call a function $v: D \rightarrow \mathbb R $
\emph{essentially $n$-dimensional} if there exist a function $\zeta: \mathbb R^n \to \mathbb R$ and a matrix $\mmat N \in \mathbb R^{ n \times d}$ with orthogonal rows such that 
$v: \mathbb R^d \rightarrow \mathbb R$ is given by
\begin{align*}
v \parent{\vvec x} &= \zeta \parent{\mmat N \vvec x}. 
\end{align*}
\end{definition}

\begin{definition}
By extension, we call a differential operator $\mathcal K$ 
essentially $n$-dimensional if the following backward PDE is well posed
\begin{align}
\label{eq:opessndim}
\begin{split}
- \upartial_t w \parent{t,\vvec x} &= \mathcal K w \parent{t,\vvec x},~~ \parent{t,\vvec x} \in [0,T) \times D,
\\
w \parent{T,\cdot } 
&=
w_T \parent{\cdot}
,
\end{split}
\end{align}
and it has an unique essentially $n$-dimensional solution for any essentially $n$-dimensional terminal value, $w_T$.
Here we specifically mean that the function $\zeta$ may depend on time, that is
\begin{align*}
w \parent{t,\vvec x} &= \zeta \parent{t,\mmat N \vvec x}, 
\end{align*}
but the matrix $\mmat N$ does not.
\end{definition}

\begin{remark}[Time independence of lower dimensional subspaces] The definition above rules out solutions to \eqref{eq:opessndim} that are essentially lower-dimensional in each instant of time although the directions along which such functions have non-vanishing partial derivatives change over time. We also tacitly assume in this definition that the allowed terminal values make the problem \eqref{eq:opessndim} well posed. 
We later exploit this structure when proving Lemma \ref{lem:european} that allows us to reduce the analysis of essentially low-dimensional models to the study of European value functions only, disregarding the possibility for early exercise.
\end{remark}

\subsubsection{Bachelier model}

First, we prove that the Markovian projection gives exact results even for American options pricing when used on the Bachelier model.
This arises from the fact that the Markovian-projected basket $\pport$ coincides in law with the true basket $\pone \vvec X $.
After discussing the one-dimensional nature of the Bachelier model, we propose possible extensions introducing a stochastic clock.

\begin{lemma}[Dimension reduction in the Bachelier model]
\label{lem:bachelierReduction}
Let $\eft t$ solve \eqref{eq:firstFullSDE} and the drift and volatility be given by \eqref{eq:riskNeutralDrift} and \eqref{eq:BachelierDefinition} respectively.
Furthermore, let $\pportt t$ be the Markovian projection defined by eqs. \eqref{eq:surrogate}, \eqref{eq:baradef} and \eqref{eq:barbdef} for $\poneft t$. Then $\pportt t$ and $\poneft t$ coincide in law.
\end{lemma}

\begin{samepage}
\begin{proof}
The proof is direct.
\end{proof}
\end{samepage}

We have established that the multivariate Bachelier model has an essentially one-dimensional generator.

However, we know that the model does not feature fat-tailed distribution for returns or clustering of volatility.
Both features have been observed in the markets (see \citet{fama1965behavior} \citet{melino1991pricing}, \citet{mandelbrot1997variation} and \citet{cont2001empirical}).
In the following Corollary, we address these issues through the introduction of a stochastic clock.
In this way, we introduce a larger class of arbitrage-free dynamics for which the price distribution conditioned to the value of the stochastic clock reduces to the one from the Bachelier model.

\begin{samepage}
\begin{corollary}[Stochastic time change in the Bachelier model]
\label{cor:stochclock}
Let $\efvec$ be given by the Bachelier model (\ref{eq:riskNeutralDrift}, \ref{eq:BachelierDefinition}) and let $U \parent t $ be an almost surely increasing process, or a stochastic clock, independent of $\efvec$ in $\mathbb R^+$, with $U \parent 0 = 0$.
Let the discounted process corresponding to $\efvec$ be
\begin{align}\label{eq:RX}
\vvec \gmartingale _{\efvec } \parent t = \discountzero t \efvec \parent t,
\end{align}
then a related stock price process $\vvec Y$, given by
\begin{align*}
\vvec Y \parent t = \discountzeroreverse t \discountzero{U \parent t} \efvec \parent{U \parent t} = \discountzeroreverse t \vvec \gmartingale_{\efvec } \parent{U \parent t}
\end{align*}
has an essentially one-dimensional generator.
\end{corollary}
\end{samepage}

\begin{samepage}
\begin{proof} The proof is divided into two steps.

\paragraph{Step 1}

The combination of \eqref{eq:RX} and \eqref{eq:firstFullSDE} yields that $\vvec \gmartingale _{\efvec }$ is a martingale with respect to its canonical filtration.
We show that the same holds for $\vvec \gmartingale_{\vvec Y} \parent t = \exp(-rt)\vvec Y \parent t = \vvec \gmartingale _{\efvec } (U(t))$.

We take $0\le s < t \le T$ and consider the conditional expectation
\begin{align*}
\expp{\vvec \gmartingale_{\vvec Y} \parent t | \vvec \gmartingale_{\vvec Y} \parent s } 
= 
&\expp{\vvec \gmartingale_{\efvec} \parent {U(t)} | \vvec \gmartingale_{\efvec} \parent { U(s)} } \\
=& 
\expp{\expp{\vvec \gmartingale_{\efvec} \parent {U(t)} | \vvec \gmartingale_{\efvec} \parent { U(s)}, U(t),U(s)} | \vvec \gmartingale_{\efvec} \parent { U(s)} } \\
=& \expp{\vvec \gmartingale _{\efvec} \parent{U \parent s} | \vvec \gmartingale _{\efvec} \parent {U \parent s}}\\
=& \vvec \gmartingale_{\vvec Y} \parent {s}. 
\end{align*}

\paragraph{Step 2} Verify the claim of essentially one-dimensionality.


Our goal now is to represent the European option price on the basket $\pone \vvec Y(T) $, $w$, in terms of
a weighted average of European options, each of them written on the basket $\pone \vvec X.$

We have, recalling that $\vvec Y(t) = \exp(rt) \exp(-rU(t)) \vvec X(U(t))$,

\begin{align}
\label{eq:Yoption}
\begin{split}
w \parent{t, \vvec y} =& \exp(-r(T-t))\expp{g \parent{\pone \vvec Y \parent T } | \vvec Y( t) = \vvec y }
\\
=& \exp(-r(T-t))\expp{\expp{g \parent{\pone \vvec Y \parent T } |U(T), U(t), \vvec Y( t) }| \vvec Y( t) = \vvec y }
\\
=& \exp(-r(T-t))\expp{{\Pi} | \vvec Y( t) = \vvec y }
\end{split}
\end{align}
%
with
$$
{\Pi} = \expp{g \parent{\exp(-r(U(T) -T))\pone \vvec X(U(T) ) } |U(T), U(t), \vvec X( U(t)) 
}
$$
being the price of a European option written on the basket $\pone \vvec X$ with maturity time $U(T)$ and time to maturity $U(T) -U(t) .$
Then, due to Lemma \ref{lem:bachelierReduction},
$ {\Pi}$ is essentially one dimensional and depends only on the basket value $$\pone \vvec X(U(t)) = \exp(-r(t-U(t))) \pone \vvec y,$$ namely
\begin{equation}\label{eq:mathH}
{\Pi} = h(\pone \vvec y,U(t)-t,U(T)-T).
\end{equation}
The combination of \eqref{eq:Yoption} and \eqref{eq:mathH} thus implies that 
$$w \parent{t, \vvec y} = \exp(-r(T-t))\expp{h(\pone \vvec y,U(t)-t,U(T)-T) } ,$$
meaning that $w$ only depends on $\pone \vvec y$, which is what we wanted to prove.
\end{proof}
\end{samepage}

\begin{remark}[On the generality of the Stochastic Clock]
We note that in proving Corollary \ref{cor:stochclock}, we allow the stochastic clock $U$ to be quite general.

However, we note that for stochastic clocks with discontinuous trajectories, the dynamics of $\vvec Y$ becomes discontinuous and thus the Gy\"ongy lemma no longer holds.
An example of $U$ with continuous trajectories is simply 
\begin{align*}
\ddif U \parent{t} =& \parent{c + V^2 \parent{t} } \ddif t , \\
U(0) =& 0
\end{align*}both
where $c>0$ and $V$ is a one-dimensional Ornstein-Uhlenbeck process.
\end{remark}

\begin{remark}[On the density of Bachelier model augmented by stochastic clock]
In the preceding discussion above, we have assumed the density of the forward process to be known.
For most choices of the stochastic clock process, this assumption will be violated.
However, we still have access to the density conditioned on the value of the stochastic clock process.
As a result, one may still evaluate the value of the projected volatility, introducing one additional quadrature and integrating over the possible values of the stochastic process.
\end{remark}

\begin{remark}[Stochastic interest rates]
\label{rem:stochIntRates}
For time dependent, stochastic interest rates independent of the price process, one may adopt essentially the same procedure as for the stochastic clock in Corollary \ref{cor:stochclock}, essentially averaging over possible values for the independent interest rate process.
\end{remark}

For other models, such as the Black-Scholes model, there is no guarantee that Markovian projection method for pricing American basket options is exact.
However, the similarity of the Black-Scholes and Bachelier models has been pointed out in the simpler European setting in earlier works by Teichmann and others. \citep{schachermayer2008close, grunspan2011note,thomson2016option}

\subsubsection{Other models in reduced dimension}

We have demonstrated that the value function of an American basket option depends only on time and one state variable in the Bachelier model.
Here, we present some particular cases in which this property holds for a more general stochastic model.
We first show that the reducibility in dimension is a phenomenon, that arises purely from the dynamics of the system, not the early exercise property of the option.

Using this result, we characterize certain parametrizations of the Black-Scholes model that reproduce the reduced dimension behavior familiar from the Bachelier model discussed in the preceding section.

\begin{lemma}[Decoupling of dimension reduction and early exercise]
\label{lem:european}
If a $d$-dimensional SDE has a generator $\mathcal L$ that is essentially one dimensional,
then the corresponding backward operator, $\mathcal H$, for the American value function,
\begin{align*}
\parent{\mathcal H v} \parent{t, \vvec x} =&
\parent{\mathcal L v } \parent{t, \vvec x} 
 \mathbf 1_{\mathoper{max} \parent{ \parent{\mathcal L v}\parent{t,\vvec x}, ~ v \parent{t, \vvec x} - g \parent{\vvec x}} > 0}
& \parent{t,\vvec x} \in [0,T] \times D
\end{align*}
is essentially one dimensional.
\end{lemma}

\begin{proof}
First, define a coordinate rotation, $\mmat Q$, $
\transpose{\mmat Q} \mmat Q = \mathbf 1$, such that the portfolio value is given by the first coordinate in the transformed coordinates $\vvec y = \mmat Q \vvec x$, with $\mmat Q$ chosen so that the first row of $\mmat Q $ 
and $\transpose{\pone}$ are collinear. In these coordinates, denote the Black-Scholes equation
for the European value function as
\begin{align}
\begin{split}
- \upartial_t u \parent{t,\vvec y} &= \mathcal L_y \parent{t,\vvec y} u \parent{t,\vvec y}, ~~~ 
\parent{t,\transpose{\mmat Q} \vvec y} \in [0,T]\times D,
 \\
u \parent{T,\vvec y} &= g \parent{\vvecc y 1}, 
~~~
\transpose{\mmat Q} \vvec y \in D.
\end{split}
\end{align}
To continue the proof, let us consider a Bermudan value function, $v_N$, with discrete equispaced monitoring times,
$t_j = \frac{jT}{N}$, $ 0\leq j \leq N$, which solves \citep{barraquand1995numerical}
\begin{align}
\begin{split}
- \upartial_t v_N \parent{t,\vvec y} &= \mathcal L_y \parent{t,\vvec y} v_N \parent{t,\vvec y}, ~~~ 
\parent{t,{\transpose{\mmat Q} \vvec y}} \in \parent{t_i,t_{i+1}} \times D, ~~ 0 \leq i \leq N,
 \\
v_N \parent{T,\vvec y} &= g \parent{\vvecc y 1}, ~~~ \transpose{\mmat Q} \vvec y \in D,
 \\
v_N \parent{t_i,\vvec y} &= \mathrm{max} \parent{v_N \parent{t_i^+,\vvec y},g \parent{\vvecc y 1}}, ~~~ 0 \leq i \leq N, ~~ \transpose{\mmat Q} \vvec y \in D.
\end{split}
\end{align}
The terminal value $g \parent{\vvecc y 1}$ is essentially one dimensional, and by the assumption on $\mathcal L$, we know that $v_N \parent{t,\vvec y}$ is essentially one dimensional for $t \in \parent{t_{N-1},t_N}$.
Thus, the function $v_N \parent{t_{N-1},\vvec y}$ is the maximum of two essentially one-dimensional functions that depend only on the $\vvecc y 1$ coordinate.
Therefore, we can conclude that
\begin{align}
\upartial_{\vvecc y j} v_N \parent{t_{N-1},\vvec y} =&0, ~~ &j >1,
~~ {\transpose{\mmat Q} \vvec y} \in D
\end{align}
and, by using the same argument for all the subsequent intervals $\parent{t_{i-1},t_i}$, we have that
\begin{align}\label{eq:zero_der}
\upartial_{\vvecc y j} v_N \parent{t,\vvec y} &= 0 , &j >1,~~ {\transpose{\mmat Q} \vvec y} \in D, ~~~ \forall t \in [0,T].
\end{align}
The American option value function, $v$, solves
\begin{align*}
- \upartial_t v \parent{t,\vvec y} &= \mathcal H_y \parent{t,\vvec y} v \parent{t,\vvec y}, ~~~ 
&\parent{t,{\transpose{\mmat Q} \vvec y}} \in [0,T]\times D,
\nonumber \\
v \parent{T, \vvec y} &= g \parent{\vvecc y 1}, ~~ &{\transpose{\mmat Q} \vvec y} \in D,
\end{align*}
where $\mathcal H_{\vvec y}$ is the $\vvec y$-coordinate representation of the operator defined in \eqref{eq:hamiltonian}.
$v$ is given as the limit of Bermudan value functions as the number of exercising times, $N$, tends to infinity:
\begin{align}\label{eq:berlimit}
v \parent {t,\vvec y} &= \mathop{\mathrm{lim}}_{N \rightarrow \infty} v_N \parent{t,\vvec y}, ~~ &\parent{t,\transpose{\mmat Q} \vvec y} \in [0,T] \times D.
\end{align}
The combination of \eqref{eq:zero_der} and \eqref{eq:berlimit} yields
\begin{align*}
\upartial_{\vvecc y j} v \parent{t,\vvec y} &= 0 , ~~ &j >1, ~ \transpose{\mmat Q} \vvec y \in D, ~~~ \forall t \in [0,T],
\end{align*}
which concludes the proof.
\end{proof}

We have already seen that the Bachelier model is one example, in which the Hamiltonian operator, $\mathcal H$, is essentially one-dimensional.
Next, we proceed to other examples of stochastic models where the generator $\mathcal L$ is essentially one-dimensional, guaranteeing dimension reduction in the American option value function.

\subsubsection{Black-Scholes model}

Next, we turn our focus to the Black-Scholes model itself and examine how it behaves under Markovian projection and whether there exist parametrizations of the model that are essentially one dimensional.

First, let us state the relevant Black-Scholes PDE \eqref{eq:europeanFullPde} corresponding to the Black-Scholes model:
\begin{align}
\begin{split}
- \upartial_t w \parent{t,\vvec x} &=
\underbrace{
 -r w \parent{t,\vvec x}
 + r \ssum{i}{} \vvecc x i \upartial_{\vvecc x i} w \parent{t,\vvec x}
 + \ssum{ij}{} \mmatc \Omega i j \vvecc x i \vvecc x j \upartial^2_{\vvecc x i \vvecc x j} w \parent{t,\vvec x}
 }_{\equiv \parent{\mathcal{L}_{BS} w } \parent{t,\vvec x}}
 , ~~~ t \in [0,T], ~ \vvec x \in D^d_{BS},
 \\
 w \parent{T,\cdot} &= g \parent{\pone \cdot},
 \end{split}
 \label{eq:BSDef}
\end{align}
where the symmetric matrix, $\mmat \Omega \in \mathbb R^{d \times d}$, is understood as the quadratic form corresponding to a volatility matrix, $\mmat \Sigma \in \mathbb R^{d \times k}$, of Equation \eqref{eq:BlackScholesDefinition}, $\mmat \Omega = \frac{\mmat \Sigma \transpose{\mmat \Sigma} }{2}$. The domain is given as $D = D^d_{BS} = \mathbb R^d_+$

\begin{remark}
\label{rem:vanishingWeights}
A trivial example of a parametrization of the Black-Scholes model for which the value function is essentially one-dimensional is the case when portfolio weights vanish except for one, $\pone = [1,0,0,\dots,0,0]$.
For such a portfolio, we can write a one-dimensional PDE describing the cost-to-go function.
\end{remark} 

For an arbitrary set of portfolio weights, $\pone$, of the Black-Scholes model Remark \ref{rem:vanishingWeights} certainly does not apply.
However, we may apply a coordinate transformation to transform the portfolio weights to the particular choice in Remark \ref{rem:vanishingWeights}.
If the resulting transformed PDE is of the form \eqref{eq:BSDef}, this is sufficient to show that the value function is essentially one-dimensional.

Below, we demonstrate this and give a particular class of parametrizations, for which the transformation is possible.
For other parametrizations, we note that these parametrizations can be approximated by ones where portfolio returns are log-normal.
For a discussion of approximating the linear combination of variables from a multivariate log-normal, we refer the reader to \citet{mehta2007approximating}.

We rotate the coordinates of the Black-Scholes equation \eqref{eq:BSDef} using the coordinate transformation, $\mmat Q$, from the proof of Lemma \ref{lem:european}.
We have
\begin{align*}
\mathcal{L}_{BS,\vvec y} u \parent{t,\vvec y}
=& -r u \parent {t,\vvec y}
\nonumber \\
&+ r \ssum{ikl}{} \mmatc Q k i \mmatc Q i l \vvecc y k \upartial_{\vvecc y l} u \parent{t,\vvec y}
\nonumber \\
&+ \ssum{ijklmn}{} \mmatc \Omega i j \mmatc Q k i \mmatc Q l j \mmatc Q j m \mmatc Q i n \vvecc y {k} \vvecc y {l} \upartial^2_{\vvecc y m \vvecc y n}
u \parent{t,\vvec y}
, 
&
t \in [0,T],
~~~
\transpose{\mmat Q} \vvec y \in D^d_{BS} .
\end{align*}

Thanks to the orthogonality of the transformation matrix $\mmat Q$, the first-order operator simplifies to
\begin{align*}
\ssum{ikl}{} \underbrace{\mmatc Q k l \mmatc Q i l }_{= \delta_{ik}} \vvecc y k \upartial_{\vvecc y l} u \parent{t,\vvec y}
=
\ssum{i}{} y_i \upartial_{y_i} u \parent{t,\vvec y}.
\end{align*}
However, the transformed second-order term does not take the form given in \eqref{eq:BSDef} in the general case.
By writing in a tensorized form
\begin{align}
\ssum{ijklmn}{} \mmatc \Omega i j \mmatc Q k i \mmatc Q l j \mmatc Q j m \mmatc Q i n \vvecc y {k} \vvecc y {l} \upartial^2_{\vvecc y m \vvecc y n}
u \parent{t,\vvec y}
= {\Gamma}_{klmn} \vvecc y {k} \vvecc y {l} \upartial^2_{y_m y_n} u \parent{t,\vvec y}
\label{eq:secondOrder}
\end{align}
we have that $\ttens \Gamma$ has in general non-diagonal terms that couple $\vvecc y k$ and $\vvecc y l$ to
$\upartial^2_{\vvecc y m \vvecc y n} u$ for \mbox{$\sset{k,l} \neq \sset{m,n}$}. Another way to write the second-order term is 
\begin{align*}
\mathrm{Tr} \parent{\mmat \Omega \mathrm{diag} \parent{\transpose{\mmat Q} \vvec y} \parent{\transpose{\mmat Q} \parent{ \mmat H u}\mmat Q} \mathrm{diag} \parent{\transpose{\mmat Q} \vvec y} }
.
\end{align*}

Using this notation, we give a particular example of a class of parametrizations of the Black-Scholes model for which the second-order term has the diagonal structure such that the generator $\mathcal L_{BS}$ is essentially one-dimensional.

\begin{corollary}[Effective one-dimensionality of Black-Scholes model when the quadratic form has equal elements]
A Black-Scholes model such that the quadratic form in \eqref{eq:BSDef} satisfies $\mmatc \Omega i j =C$ for $1 \leq i,j \leq d$ has an essentially one-dimensional generator.
\end{corollary}

\begin{proof}
The proof is direct. Writing out the second-order term \eqref{eq:secondOrder} we get
\begin{align}
&\ssum{ijklmn}{} \mmatc \Omega i j \mmatc Q k i \mmatc Q l j \mmatc Q j m \mmatc Q i n \vvecc y {k} \vvecc y {l} \upartial^2_{\vvecc y m \vvecc y n} u \parent{t,\vvec y}
\nonumber \\
=& C \ssum{iklmn}{} \mmatc Q k i \underbrace{\parent{\ssum{j}{} \mmatc Q l j \mmatc Q j m }}_{\delta_{lm}} \mmatc Q i n \vvecc y {k} \vvecc y {l} \upartial^2_{\vvecc y m \vvecc y n} u \parent{t,\vvec y}
\nonumber \\
=& C \ssum{klmn}{} \delta_{lm} \underbrace{\parent{ \ssum{i}{} \mmatc Q k i \mmatc Q i n }}_{\delta_{kn}} \vvecc y {k} \vvecc y {l} \upartial^2_{\vvecc y m \vvecc y n} u \parent{t,\vvec y}
\nonumber \\
=& C \ssum{klmn}{} \delta_{kn} \delta_{lm} \vvecc y {k} \vvecc y {l} \upartial^2_{\vvecc y m \vvecc y n} u \parent{t,\vvec y}
\nonumber \\
=& C \ssum{kl}{} \vvecc y {k} \vvecc y {l} \upartial^2_{\vvecc y k \vvecc y l} u \parent{t,\vvec y}.
\nonumber
\end{align}
\end{proof}

We have demonstrated that there is a non-trivial set of parametrizations of the Black-Scholes model such that their corresponding generators $\mathcal L_{BS}$ are essentially one-dimensional.

For parametrizations that are not essentially one-dimensional, we still note that the upper and lower bounds \eqref{eq:lowerBoundEquation} and \eqref{eq:dualapprox} still hold.
However, there is no a priori reason to believe that they coincide.
In the next section, we evaluate the bound for a range of parametrizations and argue that these bounds are often close enough to get a practical estimate of the option price.
This is expected due to the Multivariate Black-Scholes model being well approximated by an univariate Black-Scholes model on the one hand and the multivariate Bachelier model on the other.
We have established above that the Markovian projection works for pricing in both the multivariate Bachlier model as well as the univariate Black-Scholes model.
We demonstrate that this property carries over to the multivariate Black-Scholes model as a good approximation.

\section{Numerical implementation}
\label{sec:implementation}

Here, we present a numerical implementation of our proposed method.
First, we describe in Section \ref{subsec:locVol} the methods used to evaluate the coefficients of the relevant PDE \eqref{eq:projectedAmericanBackwardEquation} in $\prq D$.
We briefly introduce the solution of the projected HJB equation in Section \ref{subsec:valFun} and proceed in Section \ref{subsec:forwardeuler} to describe the evaluation of the lower and upper bounds using forward-Euler Monte Carlo simulation.
We finally discuss the errors arising in the numerical methods in Section \ref{subsec:errDecomp} and apply the proposed methods to Bachelier and Black-Scholes models of relevance in Section \ref{sec:examples}

\subsection{Evaluation of local volatility}
\label{subsec:locVol}

So far, we have bypassed the issue of how to evaluate the local projected volatility $\pvola$ in \eqref{eq:barbdef}.
In this section we first describe in Section \ref{subsubsec:laplace} how we may efficiently evaluate the high-dimensional integrals involved in the definition of the projected volatility $\pvola$.
This discussion is followed by an interpolation scheme for extending pointwise evaluations of $\pvola$ into the projected domain $\prq D$ in Section \ref{subsubsec:ext-int}.

\subsubsection{Laplace approximation}
\label{subsubsec:laplace}

To approximate $u_A$ with $\prq u_A$, we must efficiently evaluate the conditional expectations \eqref{eq:baradef} and \eqref{eq:barbdef} that involve high-dimensional integrals.
For the risk-neutral case \eqref{eq:riskNeutralDrift} that is of most interest in financial applications and options pricing, the drift part will trivially project as
\begin{align*}
\pdrift \parent{t,s} =& \expp{\pone \vvec a \parent{t,\eft t} | \poneft t = s, ~ \eft 0 = \vvec x_0}
\\
=& \expp{\pone \parent{r \eft t} | \pone \eft t =s, ~ \eft 0 = \vvec x_0}
\\
 =& r s.
\end{align*}
For the volatility, $\pvola$, we employ the Laplace approximation, by essentially finding an extremal point of the relevant unimodal integrands and applying a second-order approximation around that extremal point.
Along this line, we make the following assumption.
\begin{assumption}
\label{ass:knownDensity}
We assume that the transition density from $\vvec x_0$ to $\vvec y$ $\phi \parent{\vvec y; \vvec x_0}$ $\phi : \mathbb R^d \rightarrow \mathbb R$ corresponding to the process \eqref{eq:firstFullSDE} is a smooth function for $0 \leq t \leq T$ and it is known explicitly.
\end{assumption}
The precise implementation of this approximation can be done in various ways,
but the underlying principle remains the same.
Some of these approaches allow to relax Assumption \ref{ass:knownDensity}.
Below we outline the Laplace approximation for the case where the assumption holds.
For a more detailed account of the use of Laplace approximation, we refer the reader to \citet{shun1995laplace} and \citet{goutis1999explaining}. 

Let  
\begin{align*}
\gamma \parent s = \expp{\Psi \parent{\eft t} | \poneft t = s, ~ \eft 0 = \vvec x_0},
\end{align*}
with $\Psi \parent{\eft t} \in L^2 \parent{\mathbb R}$.
Then, this conditional expectation satisfies
\begin{align}
\expp{ \Psi \parent{\eft t} \theta \parent{\poneft t}  | \eft 0 = \vvec x_0} &=
\expp{ \gamma \parent{\poneft t} \theta \parent{\poneft t} | \eft 0 = \vvec x_0},
\label{eq:auxgamm}
\end{align}
for all $\theta$ such that $ \theta \parent{\pone \cdot} \in L^2 \parent{\mathbb R}$.
Taking in \eqref{eq:auxgamm} $\theta_h \parent x = \frac{1}{h} \mathbf{1}_{2 \absval{x-s}<h}$ for $h > 0$ and letting $h \rightarrow 0^+$
the left-hand of the previous identity becomes a surface integral over a hyperplane
\begin{align*}
{\mathoper{lim}}_{h \rightarrow 0^+}
\expp{ \Psi \parent{\eft t} \theta_h \parent{\poneft t}  }
=
\int_{\pone \vvec x = s} \Psi \parent{\vvec x} \phi \parent{\vvec x ; \vvec x_0} \ddif A\parent{ \vvec x},
\end{align*}
where $\ddif A$ denotes the differential element of the hyperplane.
For the right-hand side we have similarly
\begin{align*}
{\mathoper{lim}}_{h \rightarrow 0^+}
\expp{ \gamma \parent{\poneft t} \theta_h \parent{\poneft t}}
=
\gamma \parent{s} \int_{ \pone \vvec x = s} 
 \phi \parent{\vvec x ; \vvec x_0} \ddif A \parent{ \vvec x}
.
\end{align*}
Setting $\Psi \parent \cdot = \pbbt \parent{t,\cdot}$ and solving for $\gamma \parent s$ in \eqref{eq:auxgamm}, we have
\begin{align}
\parent{\pvola}^2 \parent{t,s}
=
\frac{\int_{\mathbb R^{d-1}} \phi \parent{\vvec x \parent{\vvec z}; \vvec x_0} \parent{\pbbt} \parent{t, \vvec x \parent{\vvec z}} \ddif \vvec z}{\int_{\mathbb R^{d-1}} \phi \parent{\vvec x \parent {\vvec z} ; \vvec x_0} \ddif \vvec z},
\label{eq:laplaceIntegral}
\end{align}%
where we treat the first variable of $\vvec x$ above as the dependent variable,
\begin{align*}
\vvecc x i \parent{\vvec z} &= \vvecc z i , ~~ \forall i>1,
\\
\vvecc x 1 \parent{\vvec z} &= \parent{\vvecc P {11}}^{-1}\parent{s-\ssum{j=2}{d} \vvecc P {1j} \vvecc z j}.
\end{align*}
Emphasizing that we work in $\mathbb R^d$, rather than the possibly bounded domain $D$,
we approximate the integrals in \eqref{eq:laplaceIntegral}, using Laplace approximation.
We replace the unimodal integrands by suitable Gaussian functions centered at their maximizing configurations, $\vvec z^* \in \R^{d-1}$ and \mbox{$\vvec z^\star \in \R^{d-1}$}.

Denoting the integrand by $\expf{f}$ and exploiting the negative-definiteness of the Hessian $\mmat H f$, we may then approximate the integrand by expanding its logarithm $f$ as follows.
\begin{align}
\begin{split}
\int_{\R^{d-1}} \expf{f \parent{\vvec z} }\ddif \vvec z &\approx
 \int_{\R^{d-1}} \expf{f \parent{\vvec z^*} + \frac{ \parent{\vvec z-\vvec z^*}^{T} \parent{\mmat H f}\parent{\vvec z^*} \parent{\vvec z-\vvec z^*}}{2}}\ddif \vvec z
 \\
 &= 
 \expf{f \parent{\vvec z^*}}
 \sqrt{\frac{\parent{2 \pi}^{d-1}}{\mathoper{det} \absval{\parent{\mmat H f} \parent{z^*}}}}
 .
 \end{split}
 \label{eq:multiDimLaplace}
\end{align}
We employ the same approximation for both the denominator and the numerator of  \eqref{eq:laplaceIntegral} and get
\begin{align}
\frac{\int_{\R^{d-1}} \expf{f \parent{\vvec z} } \ddif \vvec z}{\int_{\R^{d-1}} \expf{\tilde f \parent{\vvec z} } \ddif \vvec z}
 \approx
 \expf{f \parent{\vvec z^*} - \tilde f \parent{\vvec {z}^\star}}
 \sqrt{\frac{\mathoper{det} \absval{\parent{\mmat H \tilde f} \parent{z^ \star}}}{\mathoper{det} \absval{\parent{\mmat H f} \parent{z^*}}}} \equiv \tilde b^2_1 \parent{t,s} ,
 \label{eq:fractionOfIntegrals}
\end{align}
where
\begin{align*}
\tilde f \parent{ \vvec z} &= \mathrm{log}
\parent{\phi \parent{\vvec x \parent{\vvec z}; \vvec x_0}  }
\\
 f \parent{ \vvec z} &= \mathrm{log}
\parent{\phi \parent{\vvec x \parent{\vvec z} ; \vvec x_0} \pbbt \parent{t, \vvec x \parent{\vvec z}}}
\end{align*}
and $\vvec {z}^\star$
and
$\vvec {z}^*$ are the critical points for $\tilde f$ and $f$
respectively.

In practice, the critical configurations can be found rapidly by expanding the known integrand, $f$, to second order and applying the Newton's iteration scheme,
\begin{align}
\label{eq:newton}
\vvec z^{\parent{n+1}} &= \parent{\mmat H f \parent{\vvec z^{\parent n}}}^{-1} \nabla f \parent {
\vvec z^{\parent n}} .
\end{align}
The iteration quickly converges to an extremal point, typically within a few dozens of iterations allowing fast evaluation.
Note that in the case of the Black-Scholes model, the density $\phi$ contains a quadratic term, which makes the Newton iteration very robust to the choice of initial configuration $\vvec z ^{\parent 0}$ in \eqref{eq:newton}.

We note that the approximation is rather simple for the case where the density of the process is normal or log-normal, i.e. the original process \eqref{eq:firstFullSDE} corresponds to Bachelier or Black-Scholes model.
\citet{bayer2014asymptotics} consider the CEV model using the heat kernel approximation (see, for example, \citet{yosida1953fundamental}) for the transition density. 

For numerical results on the accuracy of the Laplace approximation, we refer the reader to Appendix \ref{sec:alternates}, where the alternate choices of coordinates for the second-order expansion are discussed, along with their respective accuracies.

\subsubsection{Extrapolation-interpolation to projected domain}
\label{subsubsec:ext-int}

To solve for the projected cost-to-go function, $\prq u_A \parent{t,s}$ in \eqref{eq:projectedAmericanBackwardEquation}, we use the Laplace approximation introduced above to evaluate the projected local volatility in a few points in the domain, $\prq D$.
We extend these values to a truncated domain in which we solve the low dimensional Equation \eqref{eq:projectedAmericanBackwardEquation}. Thanks to the smooth behavior of the the projected volatility, $\pvola$, we only need a relatively low number of evaluations to achieve high accuracy. 

However, to verify that the resulting projected cost-to-go function $\prq u_A$ is indeed a good approximation of $u_A$ using the lower and upper bounds requires Monte Carlo simulation, which is typically costly compared to the solution of the projected backward problem \eqref{eq:backwardEulerScheme}.

To evaluate the projected volatility, $\pvola$, we generate a small Monte Carlo forward-Euler sample of trajectories of the original process \eqref{eq:firstFullSDE}, as $\efvec \parent{t_n, \omega_i}$,
$0 \leq t_n \leq N_t$ and
$1 \leq i \leq M$ for $M \approx 100$, and to evaluate the essential support
$[S^- \parent{t_n}, S^+ \parent{t_n}] \subset \prq D$
of the basket process that satisfies
\begin{equation}
\label{eq:envelope}
\begin{split}
S^- \parent {t_n} = \mathop{\mathrm{min}}_i \poneft {t_n, \omega_i} ~~0 \leq n \leq N_t,
~~1 \leq i \leq M,
\\
S^+ \parent {t_n} = \mathop{\mathrm{max}}_i \poneft {t_n, \omega_i} ~~0 \leq n \leq N_t,
~~1 \leq i \leq M.
\end{split}
\end{equation}
We select a few dozen points equispaced in the intervals $[S^- \parent{t_n}, S^+ \parent{t_n}] $ for each time step $t_n$ and create a polynomial fit for $\pvola$ for each of these instances of time.

\begin{remark}
We note that the projected volatility can only be reliably evaluated inside the area where the density for $\poneft t$ is not negligible.
At the most extreme case, at the initial time, the density of $\poneft 0$ focuses on a single point.
In reality, the appropriate domain for $\prq D$ has the schematic shape depicted in Figure \ref{fig:localVolaSurface}.
However, we carry out our evaluation of $\prq u_A$ in a rectangular domain $[0,T] \times \prq D$ and extrapolate the local volatility into the whole rectangle.
In carrying out the extrapolation, we set a small minimum value for $\parent{\pvola}^2$ to guarantee numerical stability in the backward solver.
\end{remark}

\begin{figure}
\begin{center}
\begin{minipage}{150mm}
\subfigure[Third-order polynomial interpolation for the three-dimensional Black-Scholes model.
Each red line corresponds to an instant of time from $0$ to $T=\frac{1}{2}$ and is obtained through regression of a corresponding set of evaluations indicated through blue crosses.]{
\resizebox*{7cm}{!}{
\input{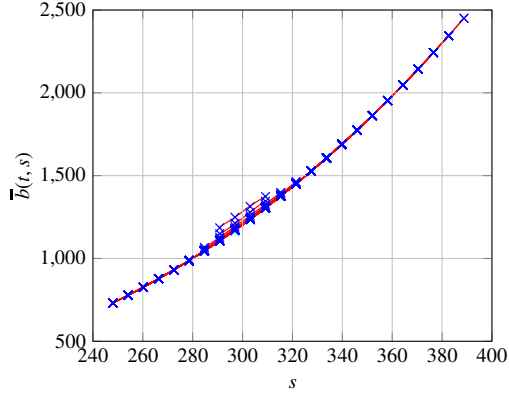}
}\label{fig:localVolaInterpolation}}
\hspace{5 mm}
\subfigure[Local volatility for the projected dynamics in the high likelihood region of the 3-to-1 dimensional example \eqref{eq:3dpars}.
For the corresponding implied volatilities, see Figure \ref{fig:impliedVolatility}]{
\resizebox*{7cm}{!}{
\input{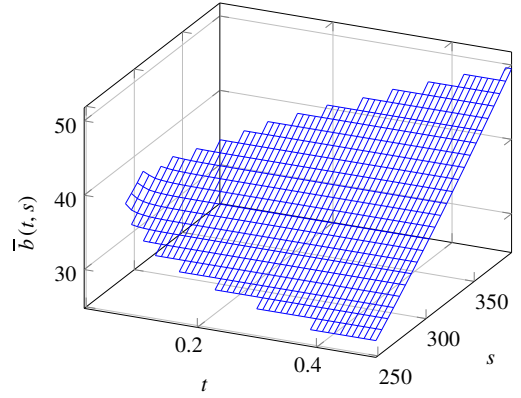}
\label{fig:localVolaSurface}
}
}
\caption{Projected volatility $\tilde b_1 \parent{t,s}$ defined in \eqref{eq:fractionOfIntegrals} and its interpolation in space and time for the 3-to-1 dimensional Black-Scholes model \eqref{eq:3dpars}.
In both the figures, the plots are done for the range of essential support of the density, which expands as $t$ increases.
\label{fig:localvola}}
\end{minipage}
\end{center}
\end{figure}

Note that the envelope \eqref{eq:envelope} is only used to get a rough estimate of where the probability mass of $\poneft {t}$ for $0\le t\le T$ lies and has a very indirect effect on the numerical solution as such.
The resulting numbers of time steps $N_t$ and samples $M$ invested in \eqref{eq:envelope} are small in comparison to the forward-Euler solution of the upper and lower bounds discussed later in Section \ref{subsec:forwardeuler}.

\subsection{Numerical value function}
\label{subsec:valFun}

Once we define the interpolated-extrapolated approximate projected volatility $\tilde b$ by interpolating the approximate projected volatility in \eqref{eq:laplaceIntegral}, we set to define a finite-difference approximation $\dbar u_A$ of the value function $\prq u_A$ that solves \eqref{eq:projectedAmericanBackwardEquation}.
Based on the finite difference operator
\begin{align*}
&\parent{\dbar {\mathcal{L}} \dbar {u}}
\parent{t,s_n}
\\
=&
\parent{\frac{\tilde b^2 \parent{t, s_n}}{2 \Delta s^2}+\frac{rs_n}{2 \Delta s} }
\dbar u \parent{t,s_{n-1}}
-
\parent{r + \frac{\tilde b^2 \parent{t,s_n}}{\Delta s^2}}
\dbar u \parent{t,s_n}
+
\parent{\frac{\tilde b^2 \parent{t, s_n}}{2 \Delta s^2}-\frac{rs_n}{2 \Delta s} }
\dbar u \parent{t,s_{n+1}}
,
\\
& 1 < n < N_s,
\end{align*}
that parallels \citep[Equation (12)]{merton1977valuation} and whose continuous counterpart is $\prq {\mathcal L}$ of \eqref{eq:theReducedEuropeanEquation}, we use a stable backward Euler scheme,
\begin{align}
\begin{split}
\dbar u_A \parent{t_{n-1}^+, s_m} &= \dbar u_A \parent{t_{n}, s_m}+ \parent{\dbar{\mathcal L} \dbar u} \parent{t_{n-1}^+, s_m} \Delta t_n, ~~~ 1\leq n \leq N_t, ~~ 1\leq m \leq N_s,
\\
\dbar u_A \parent{t_{n-1}, s_m} &= \mathoper{max} \parent {\dbar u_A \parent{t_{n-1}^+, s_m}, g \parent{s_m}},
~~~ 1\leq n \leq N_t, ~~ 1\leq m \leq N_s,
\\
\dbar u_A \parent{t_{N_t},s m} &= g \parent{s_m}
~~~ 1\leq m \leq N_s,
\end{split}
\label{eq:backwardEulerScheme}
\end{align}
with the artificial Dirichlet-type boundary condition (see Remark \ref{rem:boundary}) imposed by the payoff
\begin{align}
\begin{split}
\dbar u_A \parent{t_n,s_1} = g \parent{s_1},
\\
\dbar u_A \parent{t_n,s_{N_s}} = g \parent{s_{N_s}}
\end{split}
\end{align}
and a homogeneously spaced, time-independent, mesh $s_m = m \Delta s$.
The choice of the boundary condition has been discussed in the variational setting by \citep[pp. 316]{feng2007variational}.
The upper bound $s_{N_s}$ has to be chosen based on the magnitude of the drift and the volatility for the problem at hand.

The pointwise value function is later extended to the whole domain $\prq D$ of \eqref{eq:BSDef} using a low order interpolant, allowing the evaluation of a discrete early exercise region
\begin{align}
\prq D_{\mathrm{Ex}} = \sset{\parent{t_n, s_m}: 0\leq n\leq N_T, ~ 1\leq m \leq N_s, \dbar u_A \parent{t_n,s_m} = g \parent{s_m}}
\label{eq:approximateEarlyExerciseRegion}
\end{align}
Similarly, for the construction of the dual bound given by \eqref{eq:subOptimalMartingale}, we approximate derivatives of $ \prq u_A$ (Eq. \eqref{eq:BSDef}) using finite differences of $\dbar u_A$ (Eq. \eqref{eq:backwardEulerScheme}).

\subsection{Forward-Euler approximation}
\label{subsec:forwardeuler}

The discrete American put option value $\dbar u_A$ that solves the backward-Euler scheme \eqref{eq:backwardEulerScheme} implies a corresponding discrete early exercise region $\prq D_{\mathrm{Ex}}$ of \eqref{eq:approximateEarlyExerciseRegion}.

To verify the accuracy of the early exercise boundary implied by the discrete option value $\dbar u_A$ as an approximation to the exercise boundary in $u_A$
and to set a confidence interval for the option price, we evaluate the lower and upper bounds in Equations \eqref{eq:lowerBoundEquation} and \eqref{eq:subOptimalMartingale} using Monte Carlo simulations based on (Forward) Euler-Maruyama.
The numerical time-stepping for the asset prices, $\eft t$, is done on a uniform mesh. 
Setting the total number of time steps to coincide with the ones used in the finite difference approximation of $\dbar u_A$ defined in \eqref{eq:backwardEulerScheme}, 
avoids the need for temporal interpolation of $\dbar u_A$.
As mentioned above, we use the following discretization of \eqref{eq:firstFullSDE}:
\begin{align}
\label{eq:forwardEulerDefinition}
\begin{split}
\dbar{\efvec } \parent {t_{n+1}} =& \dbar{\efvec} \parent {t_{n}}+ r \dbar{\efvec} \parent{t_n} \Delta t_n
+
\mmat b \parent{t_{n}, \dbar{\efvec} \parent{t_n}} \Delta \vvec W \parent{t_n},
~~~ 0 \leq n < N_t,
 \\
\dbar {\efvec} \parent{t_0} =& \vvec x,
\end{split}
\end{align}
with $\Delta t_n = t_{n+1}-t_n$ and $\Delta W \parent{t_n} = W \parent{t_{n+1}}- W \parent{t_n} \sim \mathcal N \parent{0,t_{n+1}-t_n}$ and the number of time steps $N_t$.
Correspondingly, we approximate \eqref{eq:discountedPayoff} as
\begin{align}
\dbar Z \parent{t_n} =& \expf{- r t_n} 
g \parent{\pone \dbar{\efvec} \parent{t_n}},~~~
 & 0 \leq n < N_t
.
\end{align}
We use the same underlying Brownian motion to generate approximate trajectories for both the asset $\dbar \efvec$ and the approximation to the martingale $\gmartingale$ in \eqref{eq:subOptimalMartingale} used to construct the upper bound for the option price: 
\begin{align}
\begin{split}
\dbar \gmartingale \parent{t_{n+1}}
=& \dbar \gmartingale \parent{t_{n}}+ \expf{- r t_n}
\transpose{\parent{\dbar \nabla \dbar u_A}} \parent{t_n, \pone \dbar{\efvec} \parent{t_n}} \mmat b \parent{t_n, \dbar{\efvec} {t}} \Delta \vvec W \parent{t_n},
\\
&0 \leq n < N_t,
\\
\dbar \gmartingale \parent {0} =& 0.
\end{split}
\label{eq:discreteMartingaleDefinition}
\end{align}

With the discrete approximations \eqref{eq:forwardEulerDefinition} and \eqref{eq:discreteMartingaleDefinition},
we can estimate an upper bound, $A^+$, and a lower bound, $A^-$, for the option price, $u_A \parent{0, \efvec \parent 0}$, using sample averages of $M$ $i.i.d$ samples, namely
\begin{align}
\begin{split}
A^+_{M,N_t} &= \frac{1}{M}\ssum{i=1}{M} {u^+}\parent{\omega_i},
\\
u^+&= \mathop{\mathrm{max} }_{0 \leq j \leq N_t} 
\parent{\dbar Z \parent{t_j} - \dbar \gmartingale \parent{t_j}} 
\end{split}
\label{eq:maximumSimulation}
\end{align}
and
\begin{equation}
\begin{split}
A^-_{M,N_t} &= \frac{1}{M}\ssum{i=1}{M} {u^-}\parent{\omega_i},
\\
u^-&= \discountzero{\dbar \tau} g \parent{\pone \dbar {\efvec} \parent{\dbar \tau}},
\\
\upartial \prq D_{\mathrm{Ex.}} \parent{t_n} &=
\mathoper{max}
\sset{1 \leq m \leq N_s: \parent{t_n,s_m} \in \prq D_{\mathrm{Ex.}}} 
\\
\dbar \tau &= \mathop{\mathrm{min}}
\sset{0 \leq j \leq N_t: \pone \dbar {\efvec} \parent{t_j} \leq \upartial \prq D_{\mathrm{Ex.} \parent{t_n}}}.
\end{split}
\label{eq:hitTimeSimulation}
\end{equation}

To estimate the bias in the discretized approximations of the price bounds, we generate Monte Carlo samples corresponding to different values of $N_t$ and estimate the difference between the resulting estimators, 
$\absval{A^+_{M,2N_t}- A^+_{M,N_t}}$ and $\absval{A^-_{M,2N_t}- A^-_{M,N_t}}$.
For a discussion on using the forward-Euler scheme for evaluating hitting times as the one in Equation \eqref{eq:hitTimeSimulation}, we refer the reader to \citet{buchmann2003computing, bayer2010adaptive}
 
In order to accelerate the computations of the bounds, we note the possibility of using multilevel estimators instead of those in \eqref{eq:hitTimeSimulation} and \eqref{eq:maximumSimulation} 
\citep{giles2015multilevel}.
This is out of the scope of this work.

In Section \ref{sec:examples}, we present a selected set of test cases for which we evaluate the estimators \eqref{eq:maximumSimulation} and \eqref{eq:hitTimeSimulation}.
We focus in particular on the multivariate Black-Scholes that is both relevant and non-trivial and satisfies Assumption \ref{ass:knownDensity}.
The parametrizations of the Black-Scholes model we study do not feature essentially one-dimensional value functions and thus serve as a test case of our method when the accuracy of the method is not guaranteed a priori.
Still, using the lower and upper bounds, we can analyze the accuracy of our method and verify its accuracy. 
For verification purposes, we include tests on the constant-volatility Bachelier model, 
for which the Markovian projection reproduces the American option prices exactly.

\subsection{Error decomposition}
\label{subsec:errDecomp}

Before proceeding further into the numerical examples
we provide a brief summary of the errors incurred in the numerical solution of our price bounds, decomposing the total error into its constituent parts.
Denoting the estimators of \eqref{eq:maximumSimulation} and \eqref{eq:hitTimeSimulation} as
\begin{align*}
A^{\pm}_{\infty,\infty} = \mathop{\mathrm{lim}}_{M,N_t \rightarrow \infty} A^{\pm}_{M,N_t},
\end{align*}
we have that the option price $u_A$ satisfies
\begin{align*}
A^-_{\infty,\infty} \leq u_A \parent{0, \vvec x} \leq A^+_{\infty,\infty}.
\end{align*}

In practice, we rely on estimators based on finite $M$ and $N_t$. The magnitude of the gap $\absval{A^+_{\infty,\infty}-A^-_{\infty,\infty}}$ is dictated by the approximate value function $\prq u_A$ that gives rise to the inexact stopping time \eqref{eq:hittingTime} 
as well as the dual martingale $M$.
In general, finding an approximate function $\prq u_A$ that approximates the true solution $u_A$ closely might not be possible.
Furthermore, even when a sound one-dimensional approximation $\prq u_A$ exists, we rely on an approximate integration formula to recover it.
Thus, for a general model, we are not able to control the error of our method and the magnitude of the gap $\absval{A^+_{\infty,\infty}-A^-_{\infty,\infty}}$.
However, we are interested in choosing numerical parameters such that we get a reliable and useful estimate of the magnitude of this gap.

In addition to the gap between $A^+_{\infty,\infty}$ and $A^-_{\infty,\infty}$, the difference between $A^{\pm}_{\infty,\infty}$ and the corresponding estimators $A^{\pm}_{M,N_t}$ is of interest. Below, we outline the numerical approximations that give rise to these differences. Besides the fundamental error implied by approximating $\tau ^*$ of \eqref{eq:optimalStoppingTime} with $\stoptimeb$ of \eqref{eq:hittingTime},
there are four main numerical approximations employed in the procedure, with each of them giving rise to a distinct component to the error. These are:
\begin{enumerate}
\item the statistical error due to finite number of samples, $M$, in \eqref{eq:maximumSimulation} and \eqref{eq:hitTimeSimulation},
\item the step size bias introduced in the forward-Euler approximation \eqref{eq:forwardEulerDefinition},
\item the discretization errors of the solution $\prq u$, giving rise to inexact approximations to the early-exercise region and the sensitivity in \eqref{eq:discreteMartingaleDefinition}
\item the Laplace approximation error when evaluating the integrals for the coefficients of the projected dynamics and the corresponding backward solution in \eqref{eq:multiDimLaplace}.
\end{enumerate}

Noting that the choice of the time-stepping scheme implies an optimal dependence between the number of temporal and spatial discretization steps, $N_t$ and $N_s$, and using the optimal $N_s$,
we expand the notation for the estimators $A^-$ and $A^+$ to
\begin{align*}
A^{\pm}_{M,N_t} = A^{\pm}_{M,N_t,N_t,\tilde b_1},
\end{align*}
where the first $N_t$ refers to the number of forward-Euler time steps and the latter to the corresponding steps in the backward solver.
With the triangle inequality, we decompose 
\begin{align*}
\absval{A^{\pm}_{\infty,\infty}-A^{\pm}_{M,N_t}}
=&
\absval{A^{\pm}_{\infty,\infty,\infty,\pvola}-A^{\pm}_{M,N_t,N_t,\tilde b_1}}
\\
\leq&
\absval{A^{\pm}_{\infty,\infty,\infty,\pvola}-A^{\pm}_{\infty,\infty,\infty,\tilde b_1}}
+
\absval{A^{\pm}_{\infty,\infty,\infty,\tilde b_1}-A^{\pm}_{\infty,\infty,N_t,\tilde b_1}}
\nonumber
\\
&+
\absval{A^{\pm}_{\infty,\infty,N_t,\tilde b_1}-A^{\pm}_{\infty,N_t,N_t,\tilde b_1}}
+
\absval{A^{\pm}_{\infty,N_t,N_t,\tilde b_1}-A^{\pm}_{M,N_t,N_t,\tilde b_1}}.
\end{align*}

\begin{figure}
\begin{center}
\begin{minipage}{150mm}
\subfigure[
Convergence of the expected hitting time ($\tau_a$, green) to the early exercise region and the expected maximum ($X_{\mathrm{max}}$, blue) over
the interval $ 0\leq t \leq \frac{1}{2} $
for a 3-dimensional correlated Black-Scholes model \eqref{eq:3dpars} along with the $N_t^{-\frac{1}{2}}$
 reference line (dashed red).]{
\resizebox*{7cm}{!}{
%
%
%
%
\begin{tikzpicture}

\begin{loglogaxis}[
xlabel={$N_t$},
ylabel={Bias},
xmin=1, xmax=100,
ymin=0.001, ymax=100,
axis on top,
every axis plot/.append style={semithick},
xmajorgrids,
ymajorgrids,
legend entries={{$N_T^{-\frac{1}{2}}$},{$X_{\mathrm{max}}$},{$\tau_a$}}
]

\addplot [red, dashed]
coordinates {
(2,14.142135623731)
(4,10)
(8,7.07106781186548)
(16,5)
(32,3.53553390593274)
(64,2.5)
};

\addplot+[blue,mark=+,error bars/.cd, y dir=both,y explicit]
 coordinates { 
(2,9.67473741900102)    +- (0.17,0.17) 
(4,6.20904302565339) +- (0.17,0.17) 
(8,3.96238923424784) +- (0.17,0.17) 
(16,2.46719916962752)    +- (0.17,0.17) 
(32,1.36855002950341)  +- (0.17,0.17) 
(64,0.582047906118817)    +- (0.17,0.17) 
};

\addplot+[green!50.0!black,mark=x,error bars/.cd, y dir=both,y explicit]
 coordinates { 
(2,0.06501375)    +- (0.0021,0.0021) 
(4,0.04578) +- (0.0021,0.0021) 
(8,0.03070125) +- (0.0021,0.0021) 
(16,0.01951125)    +- (0.0021,0.0021) 
(32,0.011005)  +- (0.0021,0.0021) 
(64,0.00434374999999998) +- (0.0021,0.0021) 
};

\path [draw=black, fill opacity=0] (axis cs:13,1)--(axis cs:13,1);

\path [draw=black, fill opacity=0] (axis cs:1,13)--(axis cs:1,13);

\path [draw=black, fill opacity=0] (axis cs:13,0)--(axis cs:13,0);

\path [draw=black, fill opacity=0] (axis cs:0,13)--(axis cs:0,13);

\end{loglogaxis}

\end{tikzpicture}
}
\label{fig:convergences}
}
\hspace{5 mm}
\subfigure[The implied volatility for the American put option corresponding to the local volatility of the projected 3-dimensional Black-Scholes model \eqref{eq:3dpars}.
Each of the values for $\sigma_{\mathrm{imp.}}$ produces the option prices for their respective strike price, $K$, for the American option price, when the local volatility is given by the projected dynamics.]{
\resizebox*{7cm}{!}{
\input{./impliedVola_1490063127_mod.tikz}
}\label{fig:impliedVolatility}
}
\caption{}
\end{minipage}
\end{center}
\end{figure}

For the Laplace error $\absval{A^{\pm}_{\infty,\infty,\infty,\pvola}-A^{\pm}_{\infty,\infty,\infty,\tilde b_1}}$, there is no simple and practical way to control the error.
We estimate the error through the numerical experiments as presented in the appendix \ref{sec:alternates}.
All the other components are well defined and can be controlled using standard arguments in their respective numerical methods.
Firstly, with regard to the finite sample size, we can, given a confidence parameter, exploit the central limit theorem (CLT) and control the statistical error in probability by increasing the sample size,
\begin{equation}
\absval{A^{\pm}_{\infty,N_t,N_t,\tilde b_1}-A^{\pm}_{M,N_t,N_t,\tilde b_1}} = \mathcal{O}_P \parent{M^{-\frac{1}{2}}}.
\label{eq:statisticalError}
\end{equation}
As for the temporal discretization parameter, for the backward-Euler method, we set $N_s$ in \eqref{eq:backwardEulerScheme} to $N_s^2 = c N_t$, giving rise to the discretization error,
\begin{align}
\absval{A^{\pm}_{\infty,\infty,\infty,\tilde b_1}-A^{\pm}_{\infty,\infty,N_t,\tilde b_1}}
= \bigo {N_t^{-1}}.
\label{eq:backwardEulerError}
\end{align}
Finally, for the simulation of the extremal point of the dual martingale in \eqref{eq:subOptimalMartingale} and the hitting time into the early exercise region implied by $\prq u_A$, we have
\begin{align}
\absval{A^{\pm}_{\infty,\infty,N_t,\tilde b_1}-A^{\pm}_{\infty,N_t,N_t,\tilde b_1}} = \bigo{N_t^{-\frac{1}{2}}},
\label{eq:forwardSimulationError}
\end{align}
for each,
as shown in Figure \ref{fig:convergences}.

The novel contribution of this work is the use of the projected process for determining an implied exercise strategy for the true pricing problem \eqref{eq:americanPriceDefinition} using the projected value function $\prq u_A$ that solves \eqref{eq:projectedAmericanBackwardEquation}.
In the following sections, we wish to demonstrate the feasibility of this approach, and measure the resulting error, choosing parameters such that the errors \eqref{eq:statisticalError}, \eqref{eq:backwardEulerError} and \eqref{eq:forwardSimulationError} are small compared to the error implied by the use of the surrogate process and its approximate evaluation using Laplace approximation. 
We proceed to do this in the following section.

\subsection{Examples}
\label{sec:examples}

This section demonstrates the performance of our proposed method for pricing American put options written on a basket.
First, we verify our results using a $50$-dimensional Bachelier model in Section \ref{subsubsec:bacPut}.
Having verified that our numerical implementation reproduces the results expected based on Lemma \ref{lem:bachelierReduction}, we proceed to apply the method in multivariate Black-Scholes model in Sections \ref{subsubsec:3dimBS}-\ref{subsubsec:25dimBS}.

\subsubsection{American put on a basket in the Bachelier model}
\label{subsubsec:bacPut}

Here we wish to verify the numerical implementation of the finite difference solver for the approximate value function {$\dbar u_A$} of \eqref{eq:backwardEulerScheme}
and the resulting Monte Carlo estimators, \eqref{eq:hitTimeSimulation} and \eqref{eq:maximumSimulation}, for the upper and lower bounds, respectively.
We examine the solution of a 50-dimensional American put option in the Bachelier model (see Eqs. \eqref{eq:riskNeutralDrift} and \eqref{eq:BachelierDefinition}).
As our prime test case, we focus on the at-the-money put with maturity $T= \frac{1}{4}$.
To guarantee a non-trivial early exercise region, we set a relatively high interest rate of $r=0.05$.
We choose an upper diagonal $\mmat \Sigma$ with the diagonal elements $ \Sigma_{ii}=20$ for all assets $1 \leq i \leq 50$ and draw the off-diagonal components $\mmatc \Sigma i j$, $j>i$ from a standard normal distribution.

\begin{figure}
\begin{center}
\begin{minipage}{150mm}
\subfigure[The upper $A^+_{128000,N_t}$ (Blue) and lower bound $A^-_{128000,N_t}$ (Green) for the American put price for varying numbers of time steps, $N_t$, in the forward-Euler discretization.
Error bounds correspond to 95 percent confidence level. For the corresponding behavior of the relative width of the confidence interval, see Figure \ref{fig:bachelierRelativeError}.]{
\resizebox*{7cm}{!}{
\begin{tikzpicture}
\begin{semilogxaxis}[
xmin = 1, xmax = 10000,
ymin = 11, ymax = 15,
xmajorgrids,
ymajorgrids,
xlabel={$N_t$},
ylabel={Price}
]
\addplot+[green!50.0!black,mark=none,error bars/.cd, y dir=both,y explicit]
coordinates {
(8.000000,11.813338)  -= (0.0,0.047516) += (0.0, 0.047516)
(16.000000,12.279185)  -= (0.0,0.044798) += (0.0, 0.044798)
(32.000000,12.497494)  -= (0.0,0.042615) += (0.0, 0.042615)
(64.000000,12.566215)  -= (0.0,0.043522) += (0.0, 0.043522)
(128.000000,12.550167)  -= (0.0,0.042344) += (0.0, 0.042344)
(256.000000,12.675763)  -= (0.0,0.044463) += (0.0, 0.044463)
(512.000000,12.653432)  -= (0.0,0.045228) += (0.0, 0.045228)
(1024.000000,12.821465)  -= (0.0,0.045595) += (0.0, 0.045595)
(2048.000000,12.684184)  -= (0.0,0.045591) += (0.0, 0.045591)
(4096.000000,12.691837)  -= (0.0,0.045632) += (0.0, 0.045632)
};
\addplot+[blue,mark=none,error bars/.cd, y dir=both,y explicit]
coordinates {
(8.000000,14.096293)  -= (0.0,0.026023) += (0.0, 0.026023)
(16.000000,13.935821)  -= (0.0,0.017910) += (0.0, 0.017910)
(32.000000,13.609304)  -= (0.0,0.012328) += (0.0, 0.012328)
(64.000000,13.313179)  -= (0.0,0.008484) += (0.0, 0.008484)
(128.000000,13.089267)  -= (0.0,0.005895) += (0.0, 0.005895)
(256.000000,12.956044)  -= (0.0,0.004192) += (0.0, 0.004192)
(512.000000,12.857699)  -= (0.0,0.002979) += (0.0, 0.002979)
(1024.000000,12.787069)  -= (0.0,0.002100) += (0.0, 0.002100)
(2048.000000,12.735925)  -= (0.0,0.001483) += (0.0, 0.001483)
(4096.000000,12.709758)  -= (0.0,0.001056) += (0.0, 0.001056)
};
\end{semilogxaxis}
\end{tikzpicture}
}
\label{fig:convergingBounds}
}
\hspace{5 mm}
\subfigure[The distance of the error bounds relative to the underlying option price for the at-the-money put for the test case presented in Figure \ref{fig:convergingBounds}.
The estimate for the uncertainty is achieved as a combination of the upper and lower bounds presented in \ref{fig:convergingBounds}, together with an estimate of the statistical error and bias for both.]{
\resizebox*{7cm}{!}{
\begin{tikzpicture}
\begin{axis}[
xlabel={$N_t$},
ylabel={Error},
xmin=1, xmax=10000,
ymin=0.0001, ymax=1,
xmode=log,
ymode=log,
xmajorgrids,
ymajorgrids
]
\addplot [blue]
table {%
8 0.044514177670656
16 0.0444079841637766
32 0.0320127229657069
64 0.0211949732468065
128 0.0146173861352016
256 0.00981883847693009
512 0.00556450225056747
1024 0.00316459468862938
2048 0.00385915193866425
4096 0.000141037091093294
};
%
%
%
%
\end{axis}
\end{tikzpicture}
}
\label{fig:bachelierRelativeError}
}
\caption{Convergence of the upper and lower bounds for the Bachelier model described in Section \ref{subsubsec:bacPut} and the resulting relative errors for the American and at-the-money put options.
\label{fig:bachelierFigures}}
\end{minipage}
\end{center}
\end{figure}
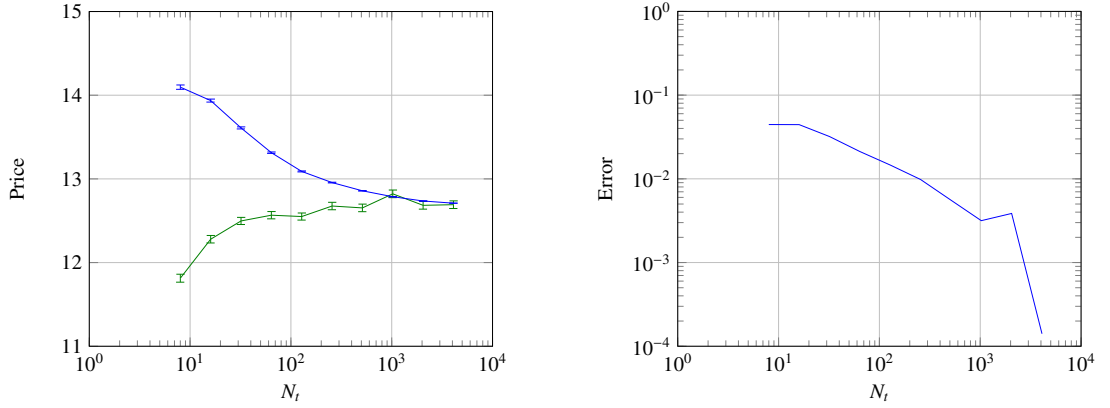

Simulating the asset dynamics, $\dbar {\efvec}$, for a sequence of time discretizations, $N_t = 1000 \times 2^{k}$, $4 \leq k \leq 11$, we observe that as $N_t$ increases, the difference between our upper and lower bounds for $u_A(0,\vvec x)$ becomes negligible.
Figure \ref{fig:convergingBounds} shows this behavior of converging bounds, alongside the statistical error of the upper bound estimator, $A^+$, which is far overshadowed by the corresponding statistical error from the lower bound estimator, $A^-$.
Indeed, as the number of time steps in the forward simulation increases, we see the upper bound intersecting the confidence interval of the lower bound,
resulting in the sub-one-percent relative error of the method.

\begin{samepage}
\subsubsection{3-to-1 dimensional Black-Scholes model}
\label{subsubsec:3dimBS}

As the first test parametrization of the Black-Scholes model we consider the case of a correlated 3-dimensional Black-Scholes model (see Eqs.\eqref{eq:riskNeutralDrift} and \eqref{eq:BlackScholesDefinition}).
We decompose the volatility function into the individual volatilities, $\volavec$, and the correlation structure of asset returns. We denote with $\mmat G$ the Cholesky decomposition of the correlation matrix of the log-returns
\begin{equation}
\mmatc \Sigma i j = \volavec _i \mmatc G i j .
\label{eq:sigmaDecomposition}
\end{equation}
We set the numerical parameters of our test case to
\begin{align}
\begin{split}
r &= 0.05, 
\\
\volavec &= \transpose{\parent{0.2, ~ 0.15, ~ 0.1}}, \\
\mmat G \transpose{\mmat G} &= \parent{\begin{matrix}
 1 & 0.8 & 0.3 \\
 0.8 & 1 & 0.1 \\
 0.3 & 0.1 & 1 
 \end{matrix}},
 \label{eq:3dpars}
 \end{split}
\end{align}
and a portfolio of equally weighted assets
\begin{align}
\pone = \arrayh{1,1,1},
\end{align}
as a representative test case of three moderately correlated assets in a high short rate environment.
The projected local volatility features noticeable skew, as shown in Figures \ref{fig:localVolaSurface} and \ref{fig:impliedVolatility}.

We evaluate the Laplace-approximated projected volatility, $\tilde b$, on a mesh of a few dozen nodes in the region where the the density of the portfolio differs significantly from zero.
Performing a regression to a third-degree polynomial on this mesh provides a close fit as seen in Figure \ref{fig:localVolaInterpolation}.
The third-order approximation also allows us to extend the evaluation of the projected volatility outside the domain in which the Laplace approximation is well behaved.
Furthermore, the coefficients of the low-order polynomial fit to the projected volatility are well approximated by a constant, or a linear function of time.
This means that for large times we can solve for the projected volatility $\pvola$ particularly sparsely in time and still have an acceptable interpolation error.

To assess the accuracy of the method, we focus on a set of put options at $T= \frac{1}{2}$ with varying moneyness and report relative numerical accuracy in the approximation of around one percent.
For the results of the prices and the corresponding relative errors, we refer to Figure \ref{fig:3dPricesAndErrors}.

\begin{figure}
\begin{center}
\begin{minipage}{150mm}
\subfigure[European (Blue) and American (Green) option prices, using forward-Euler Monte Carlo approximation and projected volatility based stopping rule and a martingale bound.]{
\resizebox*{7cm}{!}{

\begin{tikzpicture}

\begin{axis}[
ymode=log,
xlabel={$K$},
ylabel={Price},
xmin=265, xmax=335,
ymin=0.1, ymax=100,
axis on top,
xmajorgrids,
ymajorgrids
]
\addplot[green!50.0!black,mark=none,error bars/.cd, y dir=both,y explicit]
coordinates {
(270,0.63) +- (0,0.048)
(280,1.48) +- (0,0.048)
(290,3.17) +- (0,0.048)
(300,6.31) +- (0,0.048)
(310,11.55) +- (0,0.048)
(320,19.45) +- (0,0.048)
(330,30.2) +- (0,0.048)
};

\addplot[blue,mark=none,error bars/.cd, y dir=both,y explicit]
coordinates {
(270,0.50) +- (0,0.048)
(280,1.26) +- (0,0.048)
(290,2.51) +- (0,0.048)
(300,5.01) +- (0,0.048)
(310,9.25) +- (0,0.048)
(320,15.81) +- (0,0.048)
(330,25.0) +- (0,0.048)
};


\end{axis}

\end{tikzpicture}
}
\label{fig:3doptionprices}
}
\hspace{5 mm}
\subfigure[Relative errors in evaluating the American (green) and European (blue) option prices using forward-Euler Monte Carlo approximation for varying ranges of moneyness.
At high strike, $K$ we observe trivial stopping time $\pprob{\stoptimeb =0}=1$.]{
\resizebox*{7cm}{!}{

\begin{tikzpicture}

\begin{axis}[
ymode=log,
xlabel={$K$},
ylabel={Price},
xmin=265, xmax=335,
ymin=0.001, ymax=0.100,
axis on top,
xmajorgrids,
ymajorgrids
]
\addplot[blue,mark=none]
coordinates {
(270,0.035) +- (0,0.0)
(280,0.018) +- (0,0.0)
(290,0.0099) +- (0,0.0)
(300,0.006) +- (0,0.0)
(310,0.0039) +- (0,0.0)
(320,0.0028) +- (0,0.0)
(330,0.0022) +- (0,0.0)
};

\addplot[green!50.0!black,mark=none]
coordinates {
(270,0.047)
(280,0.025)
(290,0.016)
(300,0.013)
(310,0.01)
(320,0.008)
(330,0.0024)
};


\end{axis}

\end{tikzpicture}
}
\label{fig:3doptionerrors}
}
\caption{Both European and American put option prices for the test case \eqref{eq:3dpars} and the corresponding relative errors. For comparison of the solvers, identical spatial and temporal meshes, sample sizes and number of Monte Carlo realizations are used for solving both the European and the American options.
\label{fig:3dPricesAndErrors}}
\end{minipage}
\end{center}
\end{figure}
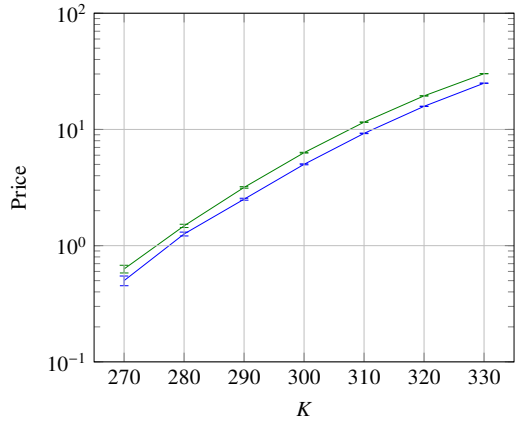
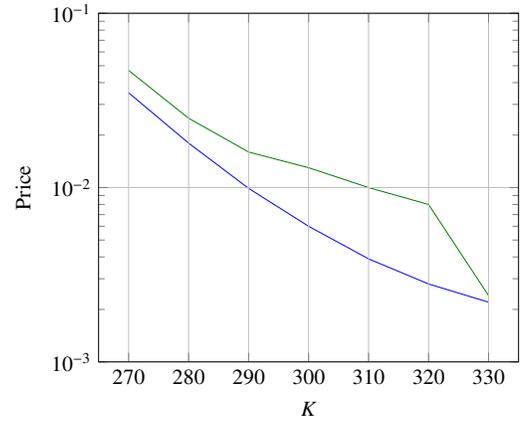

\begin{figure}
\begin{center}
\begin{minipage}{150mm}
\subfigure[Finite-difference approximation to the American value function of the 3-to-1-dimensional projected problem \eqref{eq:3dpars}.
Note that the values of the value function are used to determine an early-exercise boundary only and have no real-world interpretation except at the point $s=300$, $t=0$.]{
\resizebox*{7cm}{!}
{
\input{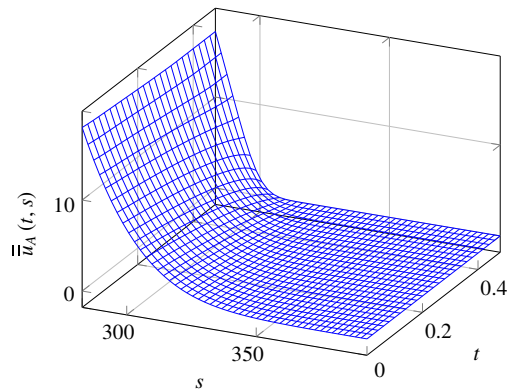}
}\label{fig:3dValueFunctionPlot} }
\hspace{5 mm}
\subfigure[Numerical finite-difference approximation of early exercise boundary of the 3-to-1-dimensional projected problem \eqref{eq:3dpars} with maturity $T=0.5$ for at-the-money put option.
A slight kink at $t<0.05$ resulting from the drop in projected volatility as seen in Figure \ref{fig:localvola} clearly visible.]{
\resizebox*{7cm}{!}{
\input{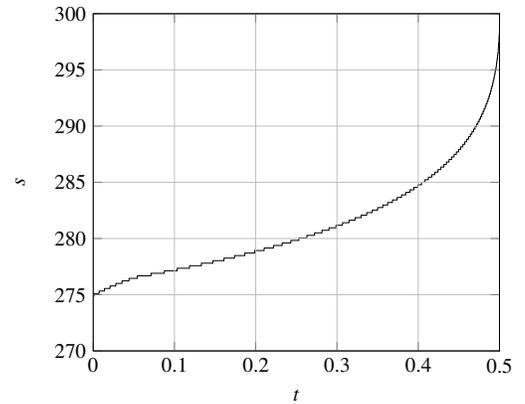}
}
\label{fig:3dEarlyExerciseBoundary}}
\caption{The value function of the 3-to-1-dimensional Black-Scholes example \eqref{eq:3dpars} and the corresponding early exercise boundary.
\label{sample-figure2}}
\end{minipage}
\end{center}
\end{figure}
\end{samepage}

\newpage

\begin{samepage}
\subsubsection{10-to-1 dimensional Black-Scholes model}
\label{subsubsec:10dimBS}

Next, we consider an example similar to \eqref{eq:3dpars}, increasing the number of dimensions to ten.
Continuing with the decomposition \eqref{eq:sigmaDecomposition}, we set
\begin{align}
\label{eq:10dpars}
\begin{split}
r &= 0.05, \\
\vvecc \sigma i &= 0.125, ~~ 1\le i\le 10 
, \\
\mmat G \transpose{\mmat G} &= \parent{\begin{matrix}
1 & 0.2 & 0.2 & 0.35 & 0.2 & 0.25 & 0.2 & 0.2 & 0.3 & 0.2 \\
0.2 & 1 & 0.2 & 0.2 & 0.2 & 0.125 & 0.45 & 0.2 & 0.2 & 0.45 \\
0.2 & 0.2 & 1 & 0.2 & 0.2 & 0.2 & 0.2 & 0.2 & 0.45 & 0.2 \\
0.35 & 0.2 & 0.2 & 1 & 0.2 & 0.2 & 0.2 & 0.2 & 0.425 & 0.2 \\
0.25 & 0.125 & 0.2 & 0.2 & 1 & 0.2 & 0.2 & 0.5 & 0.35 & 0.2 \\
0.2 & 0.45 & 0.2 & 0.2 & 0.2 & 1 & 0.2 & 0.2 & 0.2 & 0.2 \\
0.2 & 0.45 & 0.2 & 0.2 & 0.2 & 0.2 & 1 & 0.2 & 0.2 & 0.2 \\
0.2 & 0.2 & 0.2 & 0.2 & 0.2 & 0.2 & 0.2 & 1 & 0.2 & -0.1 \\
0.3 & 0.2 & 0.45 & 0.425 & 0.5 & 0.35 & 0.2 & 0.2 & 1 & 0.2 \\
0.2 & 0.45 & 0.2 & 0.2 & 0.2 & 0.2 & 0.2 & -0.1 & 0.2 & 1 
 \end{matrix}}
 .
 \end{split}
\end{align}

We evaluate a sequence of put options with varying moneyness for the equally weighted portfolio of assets namely we set $\mmatc P 1 i =1$, for all indices.
As before, we observe a relative accuracy of a few percent, with decreasing relative error as moneyness increases.
As in the previous case, with extreme moneyness, we notice the tendency for an exercise at the initial time, resulting in a variance drop of the estimators and subsequently the relative error, as shown in Figure \ref{fig:10dFigures}. The behavior of the price uncertainty of the American and European options in the 10-dimensional case, as a function of the number of time steps, $N_t$, is illustrated in Figure \ref{fig:BSRelativeErrorNt}.

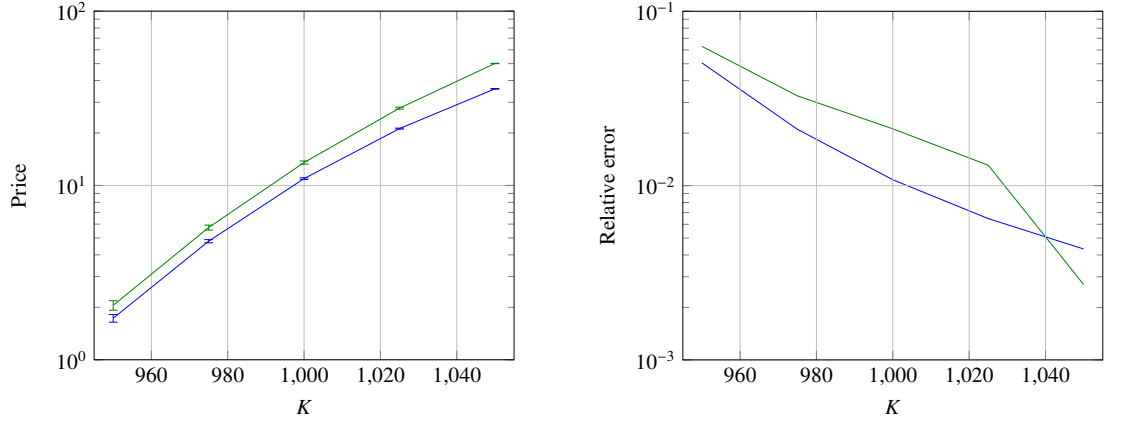
\begin{figure}
\begin{center}
\begin{minipage}{150mm}
\subfigure[European (Blue) and American (Green) put option prices for the 10-dimensional Black-Scholes test case \eqref{eq:10dpars} at $T=0.5$, using forward-Euler Monte Carlo approximation and projected volatility based stopping rule and a martingale bound.]{
\resizebox*{7cm}{!}{
\begin{tikzpicture}
\begin{semilogyaxis}[
xmin = 945, xmax = 1055,
ymin = 1, ymax = 100.0,
xlabel = {$K$},
ylabel = {Price},
xmajorgrids,
ymajorgrids
]
\addplot+[green!50.0!black,mark=none,error bars/.cd, y dir=both,y explicit]
coordinates {
(950.000000,2.053186)  -= (0.0,0.129684) += (0.0, 0.129684)
(975.000000,5.735805)  -= (0.0,0.188116) += (0.0, 0.188116)
(1000.000000,13.536625)  -= (0.0,0.286940) += (0.0, 0.286940)
(1025.000000,27.768209)  -= (0.0,0.363631) += (0.0, 0.363631)
(1050.000000,50.100563)  -= (0.0,0.135498) += (0.0, 0.135498)
};
\addplot+[blue,mark=none,error bars/.cd, y dir=both,y explicit]
coordinates {
(950.000000,1.730460)  -= (0.0,0.087534) += (0.0, 0.087534)
(975.000000,4.788305)  -= (0.0,0.100787) += (0.0, 0.100787)
(1000.000000,10.949039)  -= (0.0,0.118257) += (0.0, 0.118257)
(1025.000000,21.195846)  -= (0.0,0.137057) += (0.0, 0.137057)
(1050.000000,35.846716)  -= (0.0,0.155098) += (0.0, 0.155098)
};
\end{semilogyaxis}
\end{tikzpicture}
}
\label{fig:10doptionprices}
}
\hspace{5 mm}
\subfigure[Relative errors in evaluating the American (green) and European (blue) option prices for the 10-dimensional Black-Scholes model \eqref{eq:10dpars} using forward-Euler Monte Carlo approximation with varying ranges of moneyness.]{
\resizebox*{7cm}{!}{

\begin{tikzpicture}

\begin{semilogyaxis}[
xlabel={$K$},
ylabel={Relative error},
xmin=945, xmax=1055,
ymin=0.001, ymax=0.1,
axis on top,
xmajorgrids,
ymajorgrids
]
\addplot [blue]
coordinates {
(950,0.0505843714708641)
(975,0.0210486730847306)
(1000,0.0108006934702035)
(1025,0.00646620708246615)
(1050,0.00432669729514589)
};
\addplot [green!50.0!black]
coordinates {
(950,0.0629459282648815)
(975,0.0327278344993199)
(1000,0.0211695244991248)
(1025,0.0130856904332922)
(1050,0.002704557308929)
};
\path [draw=black, fill opacity=0] (axis cs:13,1)--(axis cs:13,1);
\path [draw=black, fill opacity=0] (axis cs:1,13)--(axis cs:1,13);
\path [draw=black, fill opacity=0] (axis cs:13,0)--(axis cs:13,0);
\path [draw=black, fill opacity=0] (axis cs:0,13)--(axis cs:0,13);

\end{semilogyaxis}

\end{tikzpicture}
}
\label{fig:10doptionerrors}
}
\caption{Convergence of the upper and lower bounds for the 10-to-1 dimensional Black-Scholes model \eqref{eq:10dpars} and the resulting relative errors. As in Figure \ref{fig:3dPricesAndErrors}, same numerical parameters have been used for solving both the European and the American prices for ease of comparison.
\label{fig:10dFigures}}
\end{minipage}
\end{center}
\end{figure}

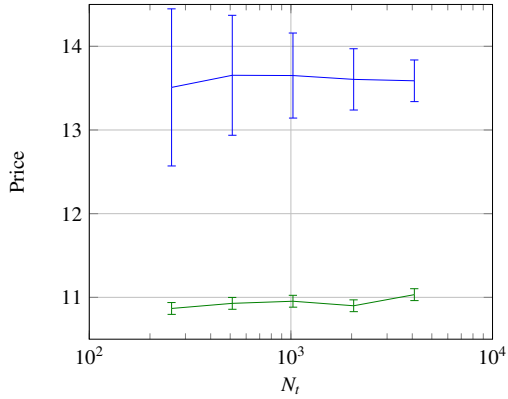
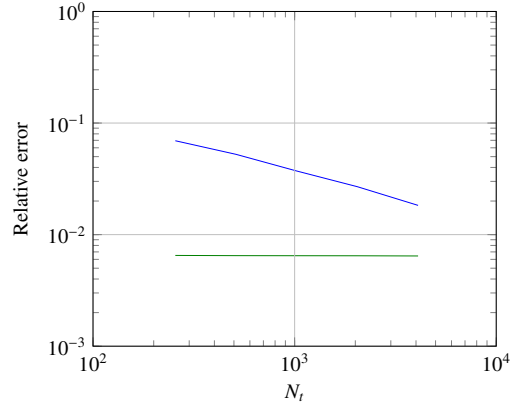
\begin{figure}
\begin{center}
\begin{minipage}{150mm}
\subfigure[Price uncertainty as a function of the number of time steps $N_t$ for the 10-dimensional Black-Scholes model for the American option (Blue) and the corresponding European option (Green) ]{
\resizebox*{7cm}{!}{
\begin{tikzpicture}
\begin{semilogxaxis}
[
xmin = 100,
xmax = 10000,
ymin = 10.5,
ymax = 14.5,
xmajorgrids,
ymajorgrids,
xlabel={$N_t$},
ylabel={Price}
]
\addplot+[blue,mark=none,error bars/.cd, y dir=both,y explicit]
coordinates {
(256.000000,13.508875)  -= (0.0,0.938716) += (0.0, 0.938716)
(512.000000,13.653647)  -= (0.0,0.716678) += (0.0, 0.716678)
(1024.000000,13.650519)  -= (0.0,0.508103) += (0.0, 0.508103)
(2048.000000,13.604665)  -= (0.0,0.366137) += (0.0, 0.366137)
(4096.000000,13.587851)  -= (0.0,0.248830) += (0.0, 0.248830)
};
\addplot+[green!50.0!black,mark=none,error bars/.cd, y dir=both,y explicit]
coordinates {
(256.000000,10.867271)  -= (0.0,0.070737) += (0.0, 0.070737)
(512.000000,10.928289)  -= (0.0,0.070821) += (0.0, 0.070821)
(1024.000000,10.953565)  -= (0.0,0.070873) += (0.0, 0.070873)
(2048.000000,10.899579)  -= (0.0,0.070449) += (0.0, 0.070449)
(4096.000000,11.032332)  -= (0.0,0.071019) += (0.0, 0.071019)
};
\end{semilogxaxis}
\end{tikzpicture}
}
\label{fig:BSAmericanEuropean}
}
\hspace{5 mm}
\subfigure[Relative error in estimating the American (Blue) And the European (Green) option price.]{
\resizebox*{7cm}{!}{
%
%
%
%
\begin{tikzpicture}

\begin{loglogaxis}[
xlabel={$N_t$},
ylabel={Relative error},
xmin=100, xmax=10000,
ymin=0.001, ymax=1,
axis on top,
xmajorgrids,
ymajorgrids
]
\addplot [green!50.0!black]
coordinates {
(256,0.00650919630900959)
(512,0.00648053365129792)
(1024,0.0064703428657088)
(2048,0.00646349356328063)
(4096,0.00643736176398458)

};
\addplot [blue]
coordinates {
(256,0.0694168981887832)
(512,0.0524305304628158)
(1024,0.0371773004833202)
(2048,0.0268785976794969)
(4096,0.0182889075943288)

};

\end{loglogaxis}

\end{tikzpicture}
\label{fig:BSRelativeErrorNt}
}
}
\caption{Price uncertainty in the 10-dimensional Black-Scholes model \eqref{eq:10dpars} with varying numbers of time steps $N_t$ in the forward-Euler monte Carlo.}
\label{fig:bsConvergences}
\end{minipage}
\end{center}
\end{figure}

\end{samepage}

\subsubsection{25-to-1 dimensional Black-Scholes model}
\label{subsubsec:25dimBS}

Finally, we consider a case with a high dimension that is certainly beyond the reach of most PDE solvers.
We choose the 25-dimensional GBM considered by \citet{bayer2016smoothing}. For the remaining parameters, we set
\begin{equation}
\begin{split}
\vvecc X i \parent 0 &= 100,    ~~~~ i \in \sset{1,2 ,\dots , 25}\\
r &= 0.05 ,
\end{split}
\label{eq:d25pars}
\end{equation}
and evaluate the options with equal portfolio weights,
$\mmatc P 1 i =1$, \mbox{$i \in \sset{1,2,\dots, 25}$}.

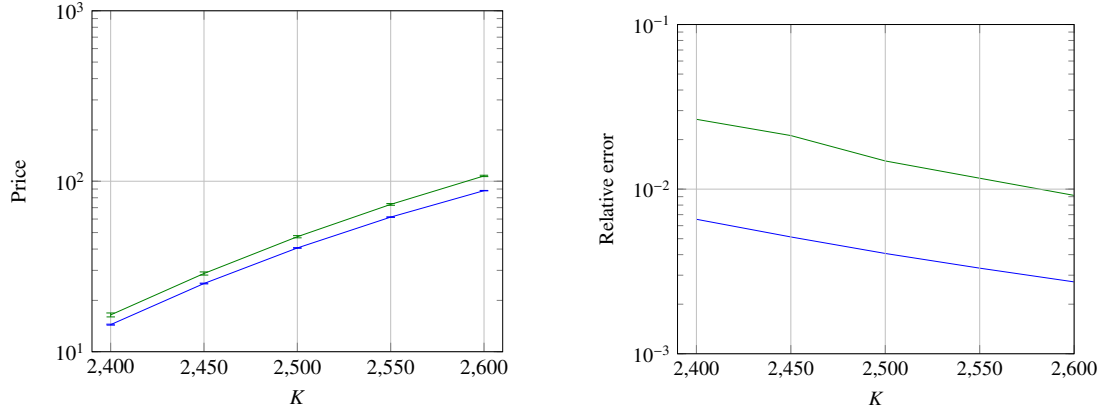
\begin{figure}
\begin{center}
\begin{minipage}{150mm}
\subfigure[European (Blue) and American (Green) put option prices for the 25-dimensional Black-Scholes test case at $T=0.5$ of Section \ref{subsubsec:25dimBS}, using forward-Euler Monte Carlo approximation and projected volatility based stopping rule and a martingale bound.]{
\resizebox*{7cm}{!}{
\begin{tikzpicture}
\begin{semilogyaxis}
[
xmin = 2390, xmax = 2610,
ymin = 10, ymax = 1000.0,
xlabel = {$K$},
ylabel = {Price},
xmajorgrids,
ymajorgrids
]

\addplot+[green!50.0!black,mark=none,error bars/.cd, y dir=both,y explicit]
coordinates {
(2400.000000,16.439245)  -= (0.0,0.437556) += (0.0, 0.437556)
(2450.000000,28.765299)  -= (0.0,0.609603) += (0.0, 0.609603)
(2500.000000,47.267012)  -= (0.0,0.701998) += (0.0, 0.701998)
(2550.000000,73.130822)  -= (0.0,0.852602) += (0.0, 0.852602)
(2600.000000,107.458191)  -= (0.0,0.984995) += (0.0, 0.984995)
};
\addplot+[blue,mark=none,error bars/.cd, y dir=both,y explicit]
coordinates {
(2400.000000,14.412621)  -= (0.0,0.094540) += (0.0, 0.094540)
(2450.000000,25.120811)  -= (0.0,0.128690) += (0.0, 0.128690)
(2500.000000,40.604863)  -= (0.0,0.165123) += (0.0, 0.165123)
(2550.000000,61.603717)  -= (0.0,0.203966) += (0.0, 0.203966)
(2600.000000,88.075997)  -= (0.0,0.240648) += (0.0, 0.240648)
};
\end{semilogyaxis}
\end{tikzpicture}
}
\label{fig:25doptionprices}
}
\hspace{5 mm}
\subfigure[Relative errors in evaluating the American (green) and European (blue) option prices for the 25-dimensional Black-Scholes model using forward-Euler Monte Carlo approximation for varying ranges of moneyness. As in earlier Figures \ref{fig:3dPricesAndErrors} and \ref{fig:10dFigures} identical numerical parameters are used for both the American and European options.]{
\resizebox*{7cm}{!}{
%
%
%
%
\begin{tikzpicture}

\begin{semilogyaxis}[
xlabel={$K$},
ylabel={Relative error},
xmin=2390, xmax=2600,
ymin=0.001, ymax=0.1,
axis on top,
xmajorgrids,
ymajorgrids
]
\addplot [blue]
coordinates {
(2400,0.0065595471656655)
(2450,0.00512285124444344)
(2500,0.00406659273968063)
(2550,0.0033109432730654)
(2600,0.00273227367237756)

};
\addplot [green!50.0!black]
coordinates {
(2400,0.0265561267505322)
(2450,0.0211558815448709)
(2500,0.0148324756307694)
(2550,0.0116474011052944)
(2600,0.00916030524703084)

};
\path [draw=black, fill opacity=0] (axis cs:13,1)--(axis cs:13,1);

\path [draw=black, fill opacity=0] (axis cs:1,13)--(axis cs:1,13);

\path [draw=black, fill opacity=0] (axis cs:13,0)--(axis cs:13,0);

\path [draw=black, fill opacity=0] (axis cs:0,13)--(axis cs:0,13);

\end{semilogyaxis}

\end{tikzpicture}
}
\label{fig:25doptionerrors}
}
\caption{Convergence of the upper and lower bounds for the 25-to-1-dimensional Black-Scholes model and the resulting relative errors.
\label{fig:25dFigures}}
\end{minipage}
\end{center}
\end{figure}

With the 25-dimensional model, we continue to observe numerical performance of a few percent of relative errors with the projected stopping rule for basket put options of maturity $T= \frac{1}{2}$ as well as a significant early-exercise premium clearly exceeding the accuracy of the method.
Results for the option price estimates for the American and European options and the corresponding error bounds are presented in Figures \ref{fig:25doptionprices} and \ref{fig:25doptionerrors}, respectively.
To demonstrate the consistency and robustness of our approach towards the particular choice of parameters, we replicate the runs multiple times with various portfolio weights.
The results of these repeated trials are illustrated in Figure \ref{fig:d25scatters}.

\begin{figure}
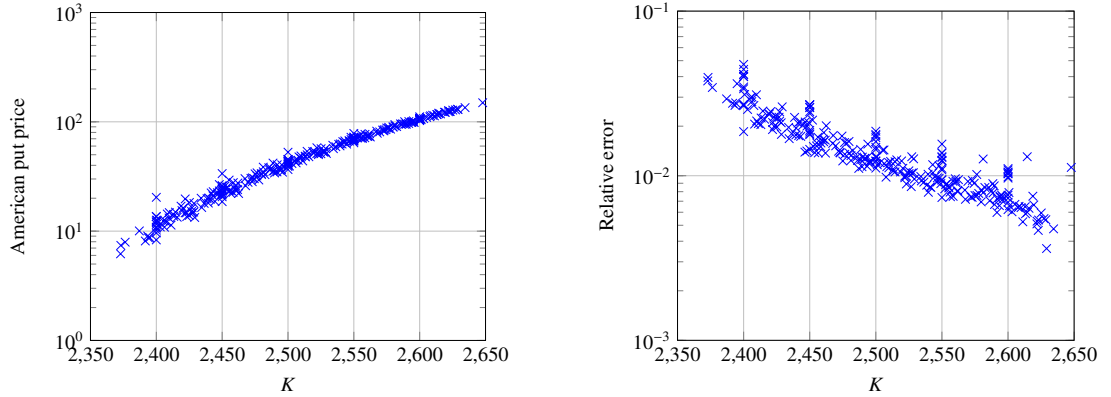

\begin{center}
\begin{minipage}{150mm}
\subfigure{
\resizebox*{7cm}{!}{
\input{./d25_scatterp_1491930791.tikz}
}
\label{fig:25dpriceScatter}
}
\hspace{5 mm}
\subfigure{
\resizebox*{7cm}{!}{
\input{./d25_scatter_1491930791.tikz}
}
\label{fig:25derrorScatter}
}
\caption{
American put prices (\ref{fig:25dpriceScatter}) and corresponding relative errors (\ref{fig:25derrorScatter}) for 43 independent randomized repetitions on evaluating the American put on the 25-to-1-dimensional Black-Scholes model with 258 individual option price valuations for varying strike, $K$. In addition to the random structure of the test problem by \citet{bayer2016smoothing} and the parameters \eqref{eq:d25pars} and $T= \frac{1}{4}$, we also randomize the portfolio weights. For each of the runs, we choose $P_{1i}$ independently from an uniform distribution $U [\frac{1}{2},\frac{3}{2}]$ and finally rescale the weights so that $\ssum{i=1}{25} P_{1i} = 25$.
\label{fig:d25scatters}}
\end{minipage}
\end{center}
\end{figure}

We note that even though we have not proven asymptotic convergence for a general multivariate model, the approximation of the true problem with the one-dimensional stopping rule gives consistently results that are comparable to the bid-ask spread of the most liquid American index options, and well below those of less-liquid regional indices and ETFs tracking them.
We also note that the relative accuracy for the American put price is greatest in the crucial region of in-the money, where the violation of the put-call parity is most profound.

\section{Conclusions}
\label{sec:conclusions}

In this work, we have demonstrated the practicability of using Markovian Projection in the framework of pricing American options written on a basket.
In the implementation of the numerical examples, we have exploited the explicitly known density of the Black-Scholes model, as well as the specific structure of the Bachelier model.
Using the known density, we devised a Laplace approximation to evaluate the volatility of a Markovian projection process that describes the projected and approximate dynamics of the basket.

We have shown that for the Bachelier model the Markovian projection gives rise to exact projected option prices, even when considering options with path-dependence.
We have also demonstrated how the vanishing derivatives of the cost-to-go function are a manifest of the process dynamics, not the early exercise nature of the option.
Leveraging this result, we have demonstrated the existence of nontrivial characterisations of the Black-Scholes model that are essentially of low dimension.

Using the Markovian projection in conjunction with the Laplace approximation, we have implemented low-dimensional approximations of various parametrizations of the multivariate Black-Scholes model.
With numerical experiments, we have shown that these approximations perform surprisingly well in evaluating prices of American options written on a basket.
We interpret these results as a manifestation of the Black-Scholes model being well approximated by a corresponding Bachelier model.
What sets these results apart from many of the earlier works is the fact that we approximate the full trajectory of a basket of assets in the Black-Scholes model, not only instantaneous returns.

The primary method used to solve such problems so far has been the least-squares Monte Carlo method that shares some common attributes with our proposed method.
Unlike least-squares Monte Carlo, our proposed method does not rely on a choice of basis vectors that are used to evaluate the holding price of an option, but only on the direction or directions along which we evaluate the projected dynamics.

Our results leave the door open for future developments including the extension of the current research into models beyond the GBM model.
We validate the accuracy of our stopping rule using a forward-simulation.
One possible extension of this work would be to use the forward sample also to evaluate the projected volatilities, an approach used in calibration of correlation structures by \citet{guyon2015cross}.
As the only non-controlled error in our method is the bias incurred in evaluating the local volatility $\pvola$, the possibility to implement such an evaluation efficiently but without introducing bias would be very useful.
From a theoretical viewpoint, our work raises the question of whether the approximation improves if the projection dimension is increased.

In this work, we have not aimed to demonstrate the use of Markovian-projected models for evaluating implied stopping times.
In doing so, we have not aimed for the greatest possible computational efficiency, and many possibilities for further optimization exist in this area.
In terms of orders of convergence, the bottleneck of the computation is the forward Euler simulation and subsequent evaluation of maxima and hitting times of realizations of an SDE.
These Monte Carlo methods could be enhanced through adaptivity, multi-level methods, use of quasi-Monte Carlo \citep{birge1994quasi,joy1996quasi}, or analytic approximations.
Likewise, there is a possibility for optimization of the numerical solver to evaluate the value function using a highly optimized backward solver \citep{khaliq2008adaptive}.
For the possibility of extending the projection to higher dimensions to allow for higher-dimensional approximation of the early exercise boundary, we refer reader to \citep{hager2010numerical}.
We also note the possibility of using a binomial tree method \citep{joshi2007convergence}, that naturally takes into account the shape of the domain $\prq D$ for the projected PDE.

We have focused on the commercially most relevant application of American options that are widely quoted on the market.
For the case of binary options the analysis remains identical, only the functional form of the payoff $g$ changes.
It would also be of interest to study the performance of the Markovian-projected dynamics in pricing other path-dependent options such as Asian and knockoff options.
Study of more general payoff functions is possible, assuming the projected volatility corresponding to these state variables could be efficiently evaluated.

\medskip

We thank Professors Ernesto Mordecki and Fabi\'an Crocce for their feedback which significantly improved this manuscript.
Gillis Danielsen provided much-valued practitioner's views.

\appendix

\section{Laplace approximation}
\label{sec:alternates}

The Taylor expansion of the integrands in equation \eqref{eq:laplaceIntegral} can be done in various ways, and we discuss and illustrate some natural choices here. For the test case, let us consider the equal-volatility, equal weight non-correlated two-dimensional Black-Scholes model with $r = 0$ and
\begin{align*}
\pone =& [1,1], 
\\
\mmat \Sigma =& \mathoper{diag} \parent{ \transpose{[\sigma , \sigma]}},
\\
\eft 0 =& \transpose{[100,100]},
\end{align*}
with the volatility, $\sigma=0.1$.
For such a simple test case, we can evaluate the relevant expansion by hand. For a high-dimensional model, we need to resort to quadratures or Monte Carlo.

Fixing the portfolio value to $\poneft 0$, the relevant unimodal integrands in terms of the natural price of the second asset $s_2$ are given as
\begin{equation}
f_1 \parent{s_2}=
\frac{\expf{
-\frac{\parent{\mathoper{log} \parent{2- \frac{s_2}{100}}}^2}{2 \sigma^2}
-\frac{\parent{\mathoper{log} \frac{s_2}{100}}^2}{2 \sigma^2}
+ 2\mathoper{log} 
\sigma
+ \mathoper{log}
\parent{
s_2^2 - 400s^2 +40000
}
- \mathoper{log} \parent{200-s_2}
- \mathoper{log} s_2
}
}{2 \pi \sigma^2}
\label{eq:volatilityInPrice}
\end{equation}
for the numerator and 
\begin{align}
\tilde f_1 \parent{s_2}
\frac{\expf{
-\frac{\parent{\mathoper{log} \parent{2- \frac{s_2}{100}}}^2}{2 \sigma^2}
-\frac{\parent{\mathoper{log} \frac{s_2}{100}}^2}{2 \sigma^2}
- \mathoper{log} \parent{200-s_2}
- \mathoper{log} s_2
}
}{2 \pi \sigma^2}
\label{eq:densityInPrice}
\end{align}
for the denominator. Alternatively, we can express the integrals in terms of log-price $x_2 = \mathoper{log} \frac{s_2}{100}$,
\begin{align}
f_2 \parent{x_2}
=
\frac{
- \frac{\parent{\mathoper{log} \parent{2- \expfs{x_2}}}^2}{2 \sigma^2}
- \frac{x_2^2}{2 \sigma^2}
+ 2 \mathoper{log} \sigma
+ \mathoper{log} \parent{2 \expfs{2 x_2} - 4 \expfs{x_2} + 4}
- \mathoper{log} \parent{2-\expfs{x_2}}
- x_2
}{2 \pi \sigma^2}
\label{eq:volatilityInLogPrice}
\end{align}
for the numerator and
\begin{align}
\tilde f_2 \parent{x_2}
=
\frac{
- \frac{\parent{\mathoper{log} \parent{2- \expfs{x_2}}}^2}{2 \sigma^2}
- \frac{x_2^2}{2 \sigma^2}
- \mathoper{log} \parent{2-\expfs{x_2}}
- x_2
}{2 \pi \sigma^2}
\label{eq:densityInLogPrice}
\end{align}
for the denominator. With these definitions we have the unit-time projected volatility
\begin{align*}
\parent{\pvola}^2 \parent{1,200} = \frac{ \int_{\mathbb R} f_1 \parent{s_2} \ddif s_2}{\int_{\mathbb R} \tilde f_1 \parent{s_2} \ddif s_2} = \frac{\int_{\mathbb R} f_2 \parent{x_2} \ddif x_2}{\int_{\mathbb R} f_2 \parent{x_2} \ddif x_2}.
\end{align*}

The integrands $f_1$, $\tilde f_1$ and their respective second-order approximations of the form $\expf{\eta+\kappa \parent{z_2-z^*}^2}$ are illustrated in Figure \ref{fig:laplace1} for the price expansion and in Figure \ref{fig:laplace2}, log-price respectively.

\begin{figure}
\begin{center}
\begin{minipage}{150mm}
\subfigure[]{
\resizebox*{7cm}{!}{
\input{./lap1d_a_1490612491.tikz}
}
\label{fig:laplacePriceDens}
}
\hspace{5 mm}
\subfigure[]{
\resizebox*{7cm}{!}{
\input{./lap1v_a_1490612491.tikz}
}
\label{fig:laplacePriceVol}
}
\caption{Functions $f_1$ of \eqref{eq:densityInPrice} (\ref{fig:laplacePriceDens}) and $\tilde f_1$ of \eqref{eq:volatilityInPrice} (\ref{fig:laplacePriceVol}) in blue and their respective approximations based on the second-order Taylor expansions of their logarithms in dashed red.
\label{fig:laplace1}
}
\end{minipage}
\end{center}
\end{figure}

\begin{figure}
\begin{center}
\begin{minipage}{150mm}
\subfigure[]{
\resizebox*{7cm}{!}{
\input{./lap2d_a_1490612491.tikz}
}
\label{fig:laplacelPriceDens}
}
\hspace{5 mm}
\subfigure[]{
\resizebox*{7cm}{!}{
\input{./lap2v_a_1490612491.tikz}
}
\label{fig:laplacelPriceVol}
}
\caption{Functions $f_2$ of \eqref{eq:densityInLogPrice} (\ref{fig:laplacelPriceDens}) and $\tilde f_2$ of \eqref{eq:volatilityInLogPrice} (\ref{fig:laplacelPriceVol}) in blue and and their respective approximations based on the second-order Taylor expansions of their logarithms in dashed red.
\label{fig:laplace2}
}
\end{minipage}
\end{center}
\end{figure}

\begin{samepage}
The approximations are given as
\begin{align*}
2 \pi \sigma^2 f_1 \parent{s_2} \approx & \expf{2 \mathoper{log} \sigma - \parent{\frac{1}{100^2 \sigma^2}+\frac{2}{100^2}} \parent{s_2-100}^2},
\\
2 \pi \sigma^2 \tilde f_1 \parent{s_2} \approx & \expf{- 2 \mathoper{log} 200 - \parent{\frac{1}{100^2 \sigma^2}+\frac{1}{100^2}} \parent{s_2-100}^2}, 
\\
2 \pi \sigma^2 f_2 \parent{x_2} \approx & \expf{ 2\mathoper{log} \sigma + \mathoper{log} 200  - \parent{\frac{1}{\sigma^2}+2} x_2^2},
\\
2 \pi \sigma^2 \tilde f_2 \parent{x_2} \approx & \expf{- \parent{\frac{1}{\sigma^2}+1} x_2^2}
,
\end{align*}
giving for both approximations
\begin{align*}
\tilde b_1^2 \parent{1,100} = \tilde b_2^2 \parent{1,100} = 20000 \sigma^2 
\sqrt{\frac{1 + 2 \sigma^2}{1+ \sigma^2}} \approx 200.99.
\end{align*}
In contrast, with quadrature, we get a reference value of $200.98$, giving a close agreement with the Laplace-approximated value.

\end{samepage}

\end{document}